\documentclass[11pt]{amsart}
\usepackage{amssymb}
\usepackage{amsmath}
\usepackage{amscd}
\newtheorem{theorem}{Theorem}[subsection]
\newtheorem{prop}[theorem]{Proposition}
\newtheorem{defn}[theorem]{Definition}
\newtheorem{lemma}[theorem]{Lemma}
\newtheorem{coro}[theorem]{Corollary}
\newtheorem{prop-def}{Proposition-Definition}[section]

\newtheorem{remark}[theorem]{Remark}

\newtheorem{exam}[theorem]{Example}

\begin{document}
\setlength{\oddsidemargin}{0cm} \setlength{\evensidemargin}{0cm}

\title{${\mathcal O}$-operators of Loday algebras and
analogues of the classical Yang-Baxter equation}

\author{Chengming Bai}

\address{Chern Institute of Mathematics \& LPMC, Nankai University,
Tannin 300071, P.R. China} \email{baicm@nankai.edu.cn}

\def\shorttitle{${\mathcal O}$-operators of Loday algebras }

\begin{abstract}

We introduce notions of ${\mathcal O}$-operators of the Loday
algebras including the dendriform algebras and quadri-algebras as a
natural generalization of Rota-Baxter operators. The invertible
$\mathcal O$-operators give a sufficient and necessary condition on
the existence of the $2^{n+1}$ operations on an algebra with the
$2^{n}$ operations in an associative cluster. The analogues of the
classical Yang-Baxter equation in these algebras can be understood
as the $\mathcal O$-operators associated to certain dual bimodules.
As a byproduct, the constraint conditions (invariances) of
nondegenerate bilinear forms on these algebras
 are given.

\end{abstract}



\maketitle

\tableofcontents \setcounter{section}{0}

\baselineskip=17pt

\section{Introduction}

Dendriform algebras are equipped with an associative product which
can be written as a linear combination of nonassociative
compositions. They were introduced by Loday (\cite{Lo1}) in 1995
with motivation from algebraic $K$-theory and have been studied
quite extensively with connections to several areas in mathematics
and physics, including operads (\cite{Lo3}), homology ([Fr1-2]),
Hopf algebras (\cite{Ch}, [H1-2], \cite{Ron}, \cite{LR2}), Lie and
Leibniz algebras (\cite{Fr2}), combinatorics (\cite{LR1}),
arithmetic(\cite{Lo2}) and quantum field theory (\cite{Fo}) and so
on (see \cite{EMP} and the references therein).

Later quite a few more similar algebra structures have been
introduced, such as quadri-algebras of Aguiar and Loday
(\cite{AL}) and octo-algebras of Leroux (\cite{Le3}). All of them
are called Loday algebras (\cite{V}, or ABQR operad algebras in
[EG1-2]). These algebras have a common property of ``splitting
associativity", that is, expressing the multiplication of an
associative algebra as the sum of a string of binary operations
(\cite{Lo3}). It is also called an (associative) cluster of
operations by Leroux ([Le1-3]). Explicitly, let $(X,*)$ be an
associative algebra over a field $\mathbb F$ of characteristic
zero and $(*_i)_{1\leq i\leq N}:X\otimes X\rightarrow X$ be a
family of binary operations on $X$. Then the operation $*$ splits
into the $N$ operations $*_1,\cdots, *_N$ or the operation $*$ is
a cluster of $N$ (binary) operations if
$$x*y=\sum_{i=1}^N x*_iy,\;\;\forall x,y\in X.\eqno (1.1)$$
In this paper, we will pay our main attention to the case that the
number $N$ of the operations in the string is $2^n$ (in fact, there
are the dendriform trialgebras with $N=3$ (\cite{E1}, \cite{LR3}),
the ennea-algebras with $N=9$ (\cite{Le1}) and so on).

Of course, besides the relation (1.1), $*_i$ should satisfy certain
additional conditions. In fact, there are several quite different
motivations to introduce the Loday algebras (including dendriform
algebras, quadri-algebras and octo-algebras) whose operation axioms
can be summarized to be a set of ``associativity" relations
(\cite{EG1}). In this paper, we give an approach which emphasizes
the bimodule (representation) structures of these algebras. For
these known algebras, one can find that in an associative cluster
the $2^{n+1}$ operations give a natural bimodule structure of an
algebra with the $2^{n}$ operations on the underlying vector space
of the algebra itself, which is the beauty of such algebra
structures. Equivalently, the ``rule'' of constructing such algebras
in an associative cluster is that, by induction, for the algebra
$(A, *_i)_{1\leq i\leq 2^n}$, besides the natural (regular) bimodule
of $A$ on
 the underlying vector space of
$A$ itself given by the left and right multiplication operators, one
can introduce the $2^{n+1}$ operations $\{*_{i_1},
*_{i_2}\}_{1\leq i\leq 2^n}$ such that
$$x*_i y=x*_{i_1} y+x*_{i_2}y,\;\;\forall x,y\in A, \; 1\leq i\leq
2^n,\eqno (1.2)$$ and their left and right multiplication operators
can give a bimodule of $(A,
*_i)_{1\leq i\leq 2^n}$ by acting on the underlying vector space of
$A$ itself. For example, it is known that a dendriform algebra gives
a natural bimodule of its associated associative algebra
(\cite{Lo1}).

These algebras have been studied extensively from many sides
including certain operadic interpretation (\cite{Le3}, \cite{EG1}).
Obviously, it is quite important to consider how to construct a new
type of algebras from the known algebras. As has been pointed out in
\cite{EG1}, one of the main themes is the use of a linear operator
with certain features. A successful example is the Rota-Baxter
operators on a known type of algebras to obtain another type of
algebras with richer structures. Rota-Baxter operators were
introduced by G. Baxter (\cite{Bax}) in 1960 and were realized their
importance by G.-C. Rota in combinatorics and other fields in
mathematics ([Rot1-4]) and have been related to several areas of
mathematics and physics (\cite{At}, [Mi1-2], \cite{Ca}, \cite{Der},
\cite{Deb},  \cite{Ng}, \cite{CK}, [E1-2], [EG1-2]). In fact, a
Rota-Baxter operator on an associative algebra can be used to
construct a dendriform algebra (\cite{Ag2}), a Rota-Baxter operator
on a dendriform algebra or a pair of commutating Rota-Baxter
operators on an associative algebra can be used to construct a
quadri-algebra (\cite{AL}), a Rota-Baxter operator on a
quadri-algebra or a set of three pairwise commutative Rota-Baxter
operators on an associative algebras can be used to construct an
octo-algebra (\cite{Le3}).

In this paper, we give some more direct relationships between these
Loday algebras. The key of our study is that we introduce a kind of
operators, namely, ${\mathcal O}$-operators of those algebras which
generalize the Rota-Baxter operators to all bimodules. In this
sense, Rota-Baxter operators are just the $\mathcal O$-operators
associated to the regular bimodules. The notion of $\mathcal
O$-operator was introduced by Kupershmidt (\cite{Ku}) for a Lie
algebra and can be traced back to Bordemann (\cite{Bo1}) in the
study of integrable systems. It was also extended to be defined for
an associative algebra (\cite{BGN} and \cite{Bai3}; such a structure
also appeared independently in \cite{U} under the name of
generalized Rota-Baxter operator).
 A direct consequence in the case of the invertible $\mathcal
O$-operators is to give a sufficient and necessary condition on the
existence of the $2^{n+1}$ operations on an algebra with the $2^{n}$
operations in an associative cluster, which to our knowledge, it has
not been  written down explicitly yet.

On the other hand, one of the reasons that Kupershmidt introduced
the $\mathcal O$-operators of a Lie algebra is to generalize (the
operator form of) the famous classical Yang-Baxter equation in the
Lie algebra (\cite{Se}, \cite{Ku}). Even an $\mathcal O$-operator of
a Lie algebra can give a solution of the classical Yang-Baxter
equation in a larger Lie algebra (\cite{Bai1}). So it is natural to
consider its analogues in associative algebras and the Loday
algebras. The associative Yang-Baxter equation was introduced by
Aguiar (\cite{Ag1}) to study the infinitesimal bialgebras given by
Joni and Rota (\cite{JR}) in order to provide an algebraic framework
for the calculus of divided difference. Its relation with ${\mathcal
O}$-operators of an associative algebra has been given in
\cite{BGN}. An analogue of the classical Yang-Baxter equation in a
dendriform algebra was introduced in \cite{Bai3} which was closely
related to a kind of bialgebra structures on dendriform algebras
(dendriform D-bialgebras). Both of these analogues can be
interpreted in terms of $\mathcal O$-operators. With an adjustable
symmetry, a solution  of the classical Yang-Baxter equation in a Lie
algebra and its analogues in an associative algebra or a dendriform
algebra is just an $\mathcal O$-operator associated to certain
``dual'' bimodules of the regular bimodules. In this sense, the
classical Yang-Baxter equation and their analogues are ``dual" to
the corresponding Rota-Baxter operators and it is also reasonable to
extend such an idea to the other Loday algebras. We believe that all
these analogues would play important roles in many fields as the
classical Yang-Baxter equation in a Lie algebra has done (\cite{Dr},
\cite{CP}).

Furthermore, the above study has a remarkable byproduct. It is known
that the nondegenerate bilinear forms satisfying certain
(``invariant") conditions are always important (\cite{Bo2}), like
the (symmetric) trace form and the (skew-symmetric) Connes cocycle
(\cite{Co}) on an associative algebra. However, sometimes, the
``invariant" conditions are not easily known for certain algebras.
Motivated by the Drinfeld's famous observation on the relation
between the invertible (skew-symmetric) solutions of the classical
Yang-Baxter equation and the 2-cocycles in a Lie algebra
(\cite{Dr}), the invertible solutions of the above analogues of the
classical Yang-Baxter equation (with an adjustable symmetry) also
can give some nondegenerate bilinear form satisfying certain
conditions. By using the $\mathcal O$-operators again, one can find
that these nondegenerate bilinear forms on an algebra with $2^n$
operations can give the splitting $2^{n+1}$ operations in an
associative cluster. We will illustrate that, starting from the
symmetric invariant (trace) forms on associative algebras, there is
a chain of bilinear forms which corresponds to the cluster of
operations.

The ${\mathcal O}$-operators of an associative algebra have been
studied in \cite{BGN} and \cite{NBG} where the relations between the
associative algebra and the compatible dendriform algebra structures
were given. In this paper, we will give a detailed study on the
other structures. We would like to point out that although many
structures are similar and many results seem ``reasonable", an
explicit and rigorous proof is still necessary. It is also necessary
to write down the details explicitly for a further development.
Furthermore, it is likely that there is an operadic interpretation
of the study in this paper. The paper is organized as follows. In
Section 2, for self-containment, we recall the relationships between
the ${\mathcal O}$-operators of associative algebras and dendriform
algebras and associative Yang-Baxter equation. In Section 3, we
study the ${\mathcal O}$-operators of dendriform algebras and their
relationships with $D$-equation which is an analogue of the
classical Yang-Baxter equation for a dendriform algebra and
quadri-algebras. In section 4, we give a similar study on the
$\mathcal O$-operator of quadri-algebras. In Section 5, we summarize
the study in the previous sections and then give a general
illustration of a similar study for any algebra with $2^n$
operations in an associative cluster.

Throughout this paper, all algebras are finite-dimensional and over
a field of characteristic zero. We also give some notations as
follows. Let $A$ be an algebra with an operation $*$.

(1) Let $L_*(x)$ and $R_*(x)$ denote the left and right
multiplication operator respectively, that is, $L_*(x)y=R_*(y)x=x*y$
for any $x,y\in A$. We also simply denote them by $L(x)$ and $R(x)$
respectively without confusion. Moreover let $L_*, R_*:A\rightarrow
gl(A)$ be two linear maps with $x\rightarrow L_*(x)$ and
$x\rightarrow R_*(x)$ respectively.

(2) Let $r=\sum_{i}a_i\otimes b_i\in A\otimes A$. Set
$$r_{12}=\sum_ia_i\otimes b_i\otimes 1, r_{13}=\sum_{i}a_i\otimes 1\otimes b_i,\;r_{23}=\sum_i1\otimes
a_i\otimes b_i,\eqno (1.3)$$where 1 is a scale. The operation
between two $r$s is in an obvious way. For example,
$$r_{12}*r_{13}=\sum_{i,j}a_i*a_j\otimes
b_i\otimes b_j,\; r_{13}* r_{23}=\sum_{i,j}a_i\otimes a_j\otimes
b_i*b_j,\;r_{23}*r_{12}=\sum_{i,j}a_i\otimes a_j* b_i\otimes
b_j,\eqno (1.4)$$ and so on.

(3) For any linear map $\rho:A\rightarrow gl(V)$, define a linear
map $\rho^*: A\rightarrow gl(V^*)$ by
$$\langle \rho(x) u, v^*\rangle=\langle u,
\rho^*(x)v^*\rangle,\;\;\forall x\in A,\; u\in V,\;v^*\in V^*,\eqno
(1.5)$$ where $\langle,\rangle$ is the ordinary pair between the
vector space $A$ and its dual space $A^*$.

\section{Preliminaries: ${\mathcal O}$-operators of associative
algebras}

For self-containment, we recall some facts on $\mathcal O$-operators
of associative algebras. Most of the results can be found in
[Ag1-4], \cite{Bai3}, \cite{BGN}, \cite{NBG}.

\subsection{${\mathcal O}$-operators of associative
algebras and dendriform algebras}

\setcounter{equation}{0}
\renewcommand{\theequation}
{2.1.\arabic{equation}}

\begin{defn}{\rm Let $A$ be an associative algebra and $V$ be a vector space.
Let $l, r:A\rightarrow gl(V)$ be two linear maps. $V$ (or the pair
$(l,r)$, or $(l,r,V)$) is called a {\it bimodule} of $A$ if
$$l(xy)v=l(x)l(y)v,\;r(xy)v=r(y)r(x)v,
\;l(x)r(y)v=r(y)l(x)v,\;\forall\; x,y\in A, v\in V.\eqno
(2.1.1)$$}\end{defn}

In fact, according to \cite{Sc}, $(l,r, V)$ is a bimodule of an
associative algebra $A$ if and only if the direct sum $A\oplus V$ of
vector spaces is turned into an associative algebra (the {\it
semidirect sum}) by defining multiplication in $A\oplus V$ by
$$(x_1+v_1)*(x_2+v_2)=x_1\cdot x_2+(l(x_1)v_2+r(x_2)v_1),\;\;\forall x_1,x_2\in A, v_1,v_2\in
V.
\eqno(2.1.2)$$ We denote it by $A\ltimes_{l,r}V$.

\begin{exam} {\rm Let $A$ be an associative algebra. If $(l, r, V)$ is a
bimodule of $A$, then $(r^*, l^*, V^*)$ is a bimodule of $A$, which
is called the dual bimodule of $(l,r,V)$. In particular, both $(L,R,
A)$ and $(R^*, L^*, A^*)$ are bimodules of $A$. The former is called
a regular bimodule of $A$.}\end{exam}

\begin{defn} {\rm Let $(A,*)$ be an associative algebra and $(l,r, V)$ be
a bimodule. A linear map $T:V\rightarrow A$ is called an {\it
${\mathcal O}$-operator associated to $(l,r,V)$} if $T$ satisfies
$$T(u)*T(v)=T(l(T(u))v+r(T(v))u),\;\;\forall u,v\in V.\eqno
(2.1.3)$$}\end{defn}

\begin{exam}{\rm  Let $(A,*)$ be an associative algebra. A linear map $R:
A\rightarrow A$ is called a Rota-Baxter operator on $A$ (of weight
zero) if $R$ is an ${\mathcal O}$-operator associated to the regular
bimodule $(L_*,R_*)$, that is, $R$ satisfies (\cite{Bax}, [Rot1-4],
etc.)
$$R(x)* R(y)=R(R(x)*y+x*R(y)),\;\;\forall\; x,y\in A.\eqno
(2.1.4)$$}\end{exam}

\begin{defn}{\rm  Let $A$ be a vector space with two
bilinear products denoted by $\prec$ and $\succ $. $(A, \prec,
\succ)$ is called a {\it dendriform algebra} if for any $x,y,z\in
A$,
$$(x\prec y)\prec z=x\prec (y*z),\;\;(x\succ y)\prec z=x\succ (y\prec
z),\;\;x\succ (y\succ z)=(x*y)\succ z,\eqno (2.1.5)$$ where
$x*y=x\prec y+x\succ y$.}\end{defn}

\begin{prop} {\rm (\cite{Lo1})}\quad Let $(A, \prec, \succ)$ be a dendriform
algebra.

(1) The product given by
$$x*y=x\prec y+x\succ y,\;\;\forall x,y\in A,\eqno (2.1.6)$$
defines an associative algebra. We call $(A,*)$ the associated
associative algebra of $(A,\succ,\prec)$ and $(A,\succ,\prec)$ is
called a compatible dendriform algebra structure on the
associative algebra $(A,*)$.

(2) $(L_\succ, R_\prec)$ is a bimodule of the associated associative
algebra $(A,*)$.\end{prop}

The following conclusion is obvious.

\begin{coro} Let $(A,*)$ be an associative algebra and $\succ,\prec$ be two
bilinear products on $A$. Then $(A,\succ,\prec)$ is a dendriform
algebra if and only if equation (2.1.6) holds and $(L_\succ,
R_\prec)$ is a bimodule of $(A,*)$.\end{coro}

\begin{theorem}{\rm (\cite{NBG})}\quad Let $T:V\rightarrow A$ be an ${\mathcal
O}$-operator of an associative algebra $(A,*)$ associated to a
bimodule $(l,r,V)$. Then there is a dendriform algebra structure on
$V$ given by
$$u\succ v=l(T(u))v,\;\;u\prec v=r(T(v))u,\;\;\forall\; u,v\in V.\eqno (2.1.7)$$
Therefore, there exists an associated associative algebra structure
on $V$ given by equation (2.1.6) and $T$ is a homomorphism of
associative algebras. Furthermore, $T(V)=\{T(v)|v\in V\}\subset A$
is an associative subalgebra of $(A,*)$ and there is an induced
dendriform algebra structure on $T(V)$ given by
$$T(u)\succ T(v)=T(u\succ v),\;\;T(u)\prec T(v)=T(u\prec v),\;\;\forall\; u,v\in V.\eqno (2.1.8)$$
Moreover, its corresponding associated associative algebra structure
on $T(V)$ given by equation (2.1.6) is just the associative
subalgebra structure of $(A,*)$ and $T$ is a homomorphism of
dendriform algebras.\end{theorem}

\begin{coro}{\rm (\cite{Ag2})}\quad Let $(A,*)$ be an associative algebra and $R$
be a Rota-Baxter operator (of weight zero) on $A$. Then there exists
a dendriform algebra structure on $A$ given by
$$x\succ y= R(x)*y,\;\;x\prec y=x*R(y),\;\;\forall x,y\in A.\eqno
(2.1.9)$$\end{coro}

\begin{coro} {\rm (\cite{NBG})}\quad Let $(A,*)$ be an associative algebra. There
is a compatible dendriform algebra structure $(\succ, \prec)$ on
$(A,*)$ if and only if there exists an invertible $\mathcal
O$-operator of $(A,
*)$.\end{coro}

In fact, if $T$ is an invertible $\mathcal O$-operator associated to
a bimodule $(l,r,V)$, then the compatible dendriform algebra
structure on $(A,*)$ is given by
$$x\succ y=T(l(x)T^{-1}(y)),\;\; x\prec y=T(r(y)T^{-1}(x)),\;\;\forall x,y\in A.\eqno (2.1.10)$$
Conversely, let $(A,\succ, \prec)$ be a dendriform algebra and
$(A,*)$ be the associated associative algebra. Then the identity map
$id:A\rightarrow A$ is an $\mathcal O$-operator  of $(A,*)$
associated to the bimodule $(L_\succ, R_\prec)$.

\subsection{${\mathcal O}$-operators of associative
algebras and associative Yang-Baxter equation}

\setcounter{equation}{0}
\renewcommand{\theequation}
{2.2.\arabic{equation}}

\begin{defn}
{\rm  Let $(A,*)$ be an associative algebra and $r\in A\otimes A$.
$r$ is called a solution of {\it associative Yang-Baxter equation}
in $A$ if $r$ satisfies
$$r_{12}*r_{13}+r_{13}*r_{23}-r_{23}*r_{12}=0.\eqno
(2.2.1)$$}\end{defn}

\begin{remark}{\rm The associative Yang-Baxter equation was introduced by  Aguiar
([Ag1-2]) with the following form  to construct an infinitesimal
bialgebra:
$$r_{13}*r_{12}-r_{12}*r_{23}+r_{23}*r_{13}=0. \eqno (2.2.2)$$
The form (2.2.1) of associative Yang-Baxter equation is exactly
equation (2.2.2) in the opposite algebra (\cite{Ag3}) and it was
introduced in \cite{Bai3} to construct a symmetric Frobenius algebra
(see the end of this subsection). In particular, when $r$ is
skew-symmetric, it is obvious that equation (2.2.1) is just equation
(2.2.2) under the operation $\sigma_{13} (x\otimes y\otimes
z)=z\otimes y\otimes x$ for any $x,y,z\in A$.}\end{remark}

Let $A$ be a vector space. For any $r\in A\otimes A$, $r$ can be
regarded as a map from $A^*$ to $A$ in the following way:
$$\langle  u^*\otimes v^*,r\rangle   =\langle  u^*,r(v^*)\rangle,\;\;\forall\; u^*,v^*\in A^*.\eqno (2.2.3)$$

\begin{prop}{\rm (\cite{Bai3})}\quad Let $(A,*)$ be an associative algebra and
$r\in A\otimes A$.  Then $r$ is a skew-symmetric solution of
associative Yang-Baxter equation in $A$ if and only if $r$ is an
${\mathcal O}$-operator of the associative algebra $(A,*)$
associated to the bimodule $(R_*^*,L_*^*)$, that is, $r$ satisfies
$$r(a^*)*r(b^*)=r(R_*^*(r(a^*))b^*+L_*^*(r(b^*))a^*),\;\;\forall\; a^*,b^*\in A^*.\eqno (2.2.4)$$
\end{prop}

In fact, there is a more general relation between $\mathcal
O$-operators of associative algebras and associative Yang-Baxter
equation. Let $\sigma:A\otimes A\rightarrow A\otimes A$ be the
exchange operator defined as
$$\sigma(x\otimes y)=y\otimes x,\;\;\forall x,y\in A.\eqno (2.2.5)$$
Furthermore, let $V_1,V_2$ be two vector spaces and
$T:V_1\rightarrow V_2$ be a linear map.  Then $T$ can be identified
as an element in $V_2\otimes V_1^*$ by $\langle T, v_2^*\otimes
v_1\rangle =\langle T(v_1), v_2^*\rangle$ for any  $v_1\in V_1,
v_2^*\in V_2^* $. Moreover, since $V_2, V_1^*$ are the subspaces of
$V_2\oplus V_1^*$, any linear map $T:V_1\rightarrow V_2$ is
obviously an element in the vector space $(V_2\oplus V_1^*)\otimes
(V_2\oplus V_1^*)$ (also see the proof of Theorem 3.3.5).

\begin{theorem} {\rm (\cite{BGN})}\quad Let $(A,*)$ be an associative algebra. Let
$(l,r,V)$ be a bimodule of $A$ and $(r^*,l^*,V^*)$ be the (dual)
bimodule given in Example 2.1.2. Let $T:V\rightarrow A$ be a linear
map identified as an element in $A\otimes V^*$ which is in the
underlying vector space of $(A\ltimes_{r^*,l^*}V^*)\otimes
(A\ltimes_{r^*,l^*}V^*)$. Then $r=T-\sigma(T)$ is a skew-symmetric
solution of associative Yang-Baxter equation in the associative
algebra $A\ltimes_{r^*,l^*} V^*$ if and only if $T$ is an ${\mathcal
O}$-operator of the associative algebra $(A,*)$ associated to the
bimodule $(l,r,V)$.\end{theorem}

\begin{defn}
{\rm  Let $(A,*)$ be an associative algebra. A bilinear form
${\mathcal B}(\;,\;)$ on $A$ is {\it invariant} (or a {\it trace
form}) if
$${\mathcal B}(x*y,z)={\mathcal B}(x,y*z),\;\;\forall\; x,y\in A.\eqno (2.2.6)$$
A  {\it Connes cocycle} of $A$ is a skew-symmetric bilinear form
$\omega$ satisfying
$$\omega (x*y,z)+\omega (y*z,x)+\omega(z*x,y)=0,\;\;\forall\; x,y,z\in
A.\eqno (2.2.7)$$}\end{defn}

\begin{theorem} {\rm ([Ag3, Proposition 2.1])}\quad Let $A$ be an associative algebra
and $r\in A\otimes A$. Suppose that $r$ is skew-symmetric and
nondegenerate. Then $r$ is a solution of associative Yang-Baxter
equation in $A$ if and only if the inverse of the isomorphism
$A^*\rightarrow A$ induced by $r$, regarded as a bilinear form
$\omega$ on $A$, is a Connes cocycle of $A$. That is,
$\omega(x,y)=\langle  r^{-1}x,y\rangle   $ for any $x,y\in
A$.\end{theorem}

\begin{coro} {\rm (\cite{Bai3})}\quad Let $(A,*)$ be an associative algebra
with a nondegenerate Connes cocycle $\omega$. Then there exists a
compatible dendriform algebra structure $\succ,\prec$ on $(A,*)$
given by
$$\omega(x\succ y,z)=\omega(y, z*x),\;\; \omega(x\prec
y,z)=\omega(x,y*z),\;\; \forall x,y,z\in A.\eqno (2.2.8)$$\end{coro}

In fact, the above dendriform algebra structure is obtained by an
invertible $\mathcal O$-operator $T$ associated to the bimodule
$(R^*,L^*)$, where $T:A^*\rightarrow A$ satisfies $\langle
T^{-1}(x), y\rangle=\omega (x,y)$ for any $x,y\in A$ .

\begin{coro} {\rm (\cite{NBG})}\quad Let $(A,\succ,\prec)$ be a dendriform
algebra. Then
$$r=\sum_{i}^n (e_i\otimes e_i^*-e_i^*\otimes e_i)\eqno (2.2.9)$$
is a solution of associative Yang-Baxter equation in the
associative algebra $A \ltimes_{R_\prec^*,L_\succ^*} A^*$, where
$\{e_1,\cdots, e_n\}$ is a basis of $A$ and $\{e_1^*,\cdots,
e_n^*\}$ is its dual basis. Moreover, there is a natural Connes
cocycle $\omega$  on the associative algebra $A
\ltimes_{R_\prec^*,L_\succ^*} A^*$ induced by $r^{-1}: A\oplus
A^*\rightarrow (A\oplus A^*)^*$, which is given by
$$\omega(x+a^*,y+b^*)=-\langle  x,b^*\rangle   +\langle  a^*,y\rangle   ,\;\;\forall x,y\in
A,\;\;a^*,b^*\in A^*.\eqno (2.2.10)$$\end{coro}

\begin{remark}{\rm The above result can be partly reformulated  from
[Ag 3, Proposition 5.8].}
\end{remark}

At the end of this subsection, we briefly introduce one of the
motivations to study the associative Yang-Baxter equation (the other
motivations can be found in [Ag1-3]), which can be regarded as a
background and explain why it is an associative analogue of the
classical Yang-Baxter equation in a Lie algebra.

Recall that a (symmetric) Frobenius algebra $(A,{\mathcal B})$ is an
associative algebra $A$ with a nondegenerate  (symmetric) invariant
bilinear form ${\mathcal B}(\;,\;)$. The study of Frobenius algebras
plays an important role in many topics in mathematics and
mathematical physics (\cite{Ko}, \cite{RFFS}, etc.). Motivated by
the study of Lie bialgebras and Manin triple (\cite{Dr}, \cite{CP}),
it is natural to consider the following construction. Let $(A,*_A)$
be an associative algebra and suppose that there is another
associative algebra structure $*_{A^*}$ on its dual space $A^*$. If
there is an associative algebra structure on the direct sum $A\oplus
A^*$ of the underlying vector spaces of $A$ and $A^*$ such that both
$A$ and $A^*$ are subalgebras and the natural bilinear form on
$A\oplus A^*$ given by
$${\mathcal B}(x+a^*,y+b^*)=\langle  x,b^*\rangle   +\langle  a^*,y\rangle   ,\;\;\forall\; a^*,b^*\in
A^*,\;x,y\in A.\eqno (2.2.11)$$ is invariant on $A\oplus A^*$, then
it is called a {\it double construction of Frobenius algebra
associated to $(A,*_A)$ and $(A^*,*_{A^*})$} and we denote it by
$(A\bowtie A^*, \mathcal B)$.

For a linear map $\phi:V_1\rightarrow V_2$, we denote the dual
(linear) map by $\phi^*:V_2^*\rightarrow V_1^*$ defined as
$$\langle  v,\phi^*(u^*)\rangle   =\langle  \phi(v),u^*\rangle   ,\;\;\forall v\in V_1,u^*\in V_2.
\eqno (2.2.12)$$ Note that $\rho^*$ given by equation (1.5) is
different with the above notation if $gl(V)$ is regarded as a vector
space, too.

\begin{defn}{\rm  Let $(A,*)$ be an associative algebra. An {\it antisymmetric
infinitesimal bialgebra structure on $A$} is a linear map
$\alpha:A\rightarrow A\otimes A$ such that

(a) $\alpha^*:A^*\otimes A^*\rightarrow A^*$ defines an associative
algebra structure on $A^*$;

(b) $\alpha$ satisfies the following two equations
$$\alpha(x* y)=(1\otimes L_* (x))\alpha
(y)+(R_*(y)\otimes 1)\alpha (y);\eqno (2.2.13)$$
$$(L_* (y)\otimes 1-1\otimes R_* (y))\alpha(x)+\sigma [(L_* (x)\otimes 1-1\otimes
R_* (x))\alpha(y)]=0,\;\;\forall x,y\in A.\eqno (2.2.14)$$ We denote
it by $(A,\alpha)$.}\end{defn}

\begin{remark}{\rm Although the notion of antisymmetric infinitesimal bialgebra
was introduced in \cite{Bai3}, in fact, it is the same structure in
the notions of associative D-bialgebra in \cite{Z} and balanced
infinitesimal bialgebra (in the opposite algebra) in
\cite{Ag3}.}\end{remark}

\begin{theorem}{\rm (\cite{Z}, \cite{Bai3})}\quad Let $(A,*_A)$ be an associative
algebra.
 Suppose there is another associative algebra structure $``*_{A^*}"$ on its dual space
$A^*$ given by a linear map $\alpha^*: A^*\otimes A^*\rightarrow
A^*$.
 Then
there exists a double construction of Frobenius algebra associated
to $(A,*_A)$ and $(A,*_{A^*})$ if and only if $(A,\alpha)$ is an
antisymmetric infinitesimal bialgebra. Moreover, every double
construction of Frobenius algebra can be obtained from the above
way.\end{theorem}

The associative algebra structure on $A\oplus A^*$ is given by (for
any $x,y\in A$ and $a^*,b^*\in A^*$)
$$(x+a^*)*(y+b^*)=x*_Ay+R^*_{*_{A^*}}(a^*)y+L^*_{*_{A^*}}(b^*)x
+a^**_{A^*}b^*+R^*_{*_A}(x)b^*+L^*_{*_A}(y)a^*.\eqno (2.2.15)$$

\begin{prop}{\rm (\cite{Ag3}, \cite{Bai3})}\quad Let $(A,*)$ be an associative
algebra and $r\in A\otimes A$. If $r$ is a skew-symmetric solution
of the associative Yang-Baxter equation in $A$, then the linear map
$\alpha$ defined by
$$\alpha (x)=(1\otimes L(x)-R(x)\otimes 1)r,\;\;\forall x\in A,\eqno (2.2.16)$$
induces an associative algebra structure on $A^*$ such that
$(A,\alpha)$ is an antisymmetric infinitesimal bialgebra.\end{prop}

\subsection{${\mathcal O}$-operators of associative algebras and
double constructions of  Connes cocycles} Motivated by the study of
Lie bialgebras (\cite{Dr}, \cite{CP}) and antisymmetric
infinitesimal bialgebras in the previous subsection, it is natural
to consider the following constructions of the nondegenerate Connes
cocycles of associative algebras (a similar study on the
nondegenerate skew-symmetric 2-cocycles of Lie algebras has been
given in \cite{Bai2}).

Let $(A,*_A)$ be an associative algebra and suppose that there is
another associative algebra structure $*_{A^*}$ on its dual space
$A^*$. If there is an associative algebra structure on the direct
sum $A\oplus A^*$ of the underlying vector spaces of $A$ and $A^*$
such that both $A$ and $A^*$ are subalgebras and the
skew-symmetric bilinear form on $A\oplus A^*$ given by equation
(2.2.10) is a Connes cocycle on $A\oplus A^*$, then it is called a
{\it double construction of Connes cocycle associated to $(A,*_A)$
and $(A^*,*_{A^*})$} and we denote it by $(T(A)=A\bowtie A^*,
\omega)$.

\begin{defn} {\rm Let $A$ be a vector space. A {\it dendriform D-bialgebra structure
on $A$} is a set of linear maps
$\alpha_\prec,\alpha_\succ,\beta_\prec,\beta_\succ$ such that
$\alpha_\prec,\alpha_\succ:A\rightarrow A\otimes A$,
$\beta_\prec,\beta_\succ: A^*\rightarrow A^*\otimes A^*$ and

(a) $(\alpha_\prec^*,\alpha_\succ^*):A^*\otimes A^*\rightarrow A^*$
defines a dendriform algebra structure $(\succ_{A^*},\prec_{A^*})$
on $A^*$;

(b) $(\beta_\prec^*,\beta_\succ^*):A\otimes A\rightarrow A$ defines
a dendriform algebra structure $(\succ_{A},\prec_A)$ on $A$;

(c) The following equations are satisfied(for any $x,y\in A$ and
$a,b\in A^*$).

(1) $\alpha_\prec(x*_Ay)=(1\otimes
L_{\prec_A}(x))\alpha_\prec(y)+(R_A(y)\otimes 1) \alpha_\prec
(x)$;\hfill (2.3.1)

(2) $\alpha_\succ (x*_Ay)=(1\otimes
L_A(x))\alpha_\succ(y)+(R_{\prec_A}(y)\otimes
1)\alpha_\succ(x)$;\hfill (2.3.2)

(3) $\beta_\prec(a*_{A^*}b)=(1\otimes
L_{\prec_{A^*}}(a))\beta_\prec(b)+(R_{A^*}(b)\otimes 1) \beta_\prec
(a)$;\hfill (2.3.3)

(4) $\beta_\succ (a*_{A^*}b)=(1\otimes
L_{A^*}(a))\beta_\succ(b)+(R_{\prec_{A^*}}(b)\otimes
1)\beta_\succ(a)$;\hfill (2.3.4)

(5) $(L_A(x)\otimes 1-1\otimes R_{\prec_A}
(x))\alpha_\prec(y)+\sigma [(L_{\succ_A}(y)\otimes -1\otimes
R_A(y))\alpha_\succ(x)]=0$;\hfill (2.3.5)

(6) $(L_{A^*}(a)\otimes 1-1\otimes R_{\prec_{A^*}}
(a))\beta_\prec(b)+\sigma [(L_{\succ_{A^*}}(b)\otimes -1\otimes
R_{A^*}(b))\beta_\succ(a)]=0$,\hfill (2.3.6)

\noindent where
$L_A=L_{\succ_A}+L_{\prec_A},\;\;R_A=R_{\succ_A}+R_{\prec_A},
L_{A^*}=L_{\succ_{A^*}}+L_{\prec_{A^*}},\;\;R_{A^*}=R_{\succ_{A^*}}+R_{\prec_{A^*}}$.

We also denote this dendriform D-bialgebra by
$(A,A^*,\alpha_\succ,\alpha_\prec,\beta_\succ,\beta_\prec)$ or
simply $(A,A^*)$.}\end{defn}

\begin{theorem} {\rm (\cite{Bai3})}\quad Let $(A,\succ_A,\prec_A)$ be a dendriform
algebra whose products are given by two linear maps $\beta_\succ^*,
\beta_\prec^* :A\otimes A\rightarrow A$ and $(A,*_A)$ be the
associated associative algebra. Suppose there is another dendriform
algebra structure $``\succ_{A^*},\prec_{A^*}"$ on its dual space
$A^*$ given by two linear maps $\alpha_\succ^*,\alpha_\prec^*:
A^*\otimes A^*\rightarrow A^*$ and $(A^*,*_{A^*})$ is the associated
associative algebra. Then  there exists a double construction of
Connes cocycle associated to $(A,*_A)$ and $(A^*,*_{A^*})$ if and
only if $(A,A^*,\alpha_\succ,\alpha_\prec,\beta_\succ,\beta_\prec)$
is a dendriform D-bialgebra. Moreover, every double construction of
Connes cocycle can be obtained from the above way.\end{theorem}

The associative algebra structure on $A\oplus A^*$ is given by (for
any $x,y\in A$ and $a^*,b^*\in A^*$)
$$(x+a^*)*(y+b^*)=x*_Ay+R^*_{\prec_{A^*}}(a^*)y+L^*_{\succ_{A^*}}(b^*)x
+a^**_{A^*}b^*+R^*_{\prec_A}(x)b^*+L^*_{\succ_A}(y)a^*.\eqno
(2.3.7)$$

\begin{prop} {\rm (\cite{Bai3})}\quad Let $(A,\succ,\prec)$ be a dendriform
algebra and $(A,*)$ be the associated associative algebra. Let
$r\in A\otimes A$ and the maps $\alpha_\succ,\alpha_\prec$ be
defined by
$$\alpha_\succ(x)=(-1\otimes
L_*(x)+R_\prec(x)\otimes 1)r;\eqno (2.3.8)$$
$$\alpha_\prec(x)=(1\otimes L_\succ(x)-R_*(x)\otimes
1)r,\eqno (2.3.9)$$ for any $x\in A$. Suppose $r$ is symmetric and
$r$ satisfies
$$r_{12}*r_{13}=r_{13}\prec r_{23}+r_{23}\succ r_{12}.\eqno
(2.3.10)$$ Then the maps $\alpha_\succ,\alpha_\prec$  induce a
dendriform algebra structure on $A^*$ such that $(A,A^*)$ is a
dendriform D-bialgebra.\end{prop}

\begin{defn} {\rm Let $(A,\succ,\prec)$ be a dendriform algebra and $r\in
A\otimes A$. Equation (2.3.10) is called {\it ${D}$-equation in
$(A,\succ,\prec)$}.}\end{defn}

\begin{prop} {\rm (\cite{Bai3})} \quad Let $(A,\succ,\prec)$ be a dendriform
algebra and $(A,*)$ be the associated associative algebra. Let $r\in
A\otimes A$. Then $r$ is a symmetric solution of $D$-equation in the
dendriform algebra $(A,\succ,\prec)$ if and only if $r$ is an
$\mathcal O$-operator of the associative algebra $(A,*)$ associated
to the bimodule $(R_\prec^*, L_\succ^*)$.\end{prop}

\section{${\mathcal O}$-operators of dendriform algebras}

\subsection{Bimodules of dendriform algebras}

\begin{defn}
{\rm (\cite{Ag4})\quad Let $(A,\succ, \prec)$ be a dendriform
algebra and $V$ be a vector space. Let $l_\succ, r_\succ,
l_\prec,r_\prec:A\rightarrow gl(V)$ be four linear maps. $V$ (or
$(l_\succ,r_\succ,l_\prec,r_\prec)$, or
$(l_\succ,r_\succ,l_\prec,r_\prec,V)$) is called a {\it bimodule}
of $A$ if the following equations hold (for any $x,y\in A$):
$$l_\prec(x\prec y)=l_\prec(x)l_\prec(y)+l_\prec(x)l_\succ(y);\eqno (3.1.1)$$
$$r_\prec(x)l_\prec(y)=l_\prec(y)r_\prec(x)+l_\prec(y)r_\succ(x);\eqno
(3.1.2)$$
$$r_\prec(x)r_\prec(y)=r_\prec(y\prec x)+r_\prec(y\succ x);\eqno
(3.1.3)$$
$$l_\prec(x\succ y)=l_\succ(x)l_\prec(y);\eqno (3.1.4)$$
$$r_\prec(x)l_\succ (y)=l_\succ(y)r_\prec(x);\eqno (3.1.5)$$
$$r_\prec(x)r_\succ(y)=r_\succ(y\prec x);\eqno (3.1.6)$$
$$l_\succ (x\prec y)+l_\succ(x\succ y)=l_\succ(x)l_\succ (y);\eqno
(3.1.7)$$
$$r_\succ(x)l_\prec(y)+r_\succ(x)l_\succ(y)=l_\succ(y)r_\succ(x);\eqno
(3.1.8)$$
$$r_\succ (x)r_\succ(y)+r_\succ(x)r_\prec(y)=r_\succ(y\succ x).\eqno
(3.1.9)$$}\end{defn}

According to \cite{Sc}, $(l_\succ, r_\succ, l_\prec, r_\prec, V)$
is a bimodule of a dendriform algebra $(A,\succ, \prec)$ if and
only if there exists a dendriform algebra structure on  the direct
sum $A\oplus V$ of the underlying vector spaces of $A$ and $V$
given by ($\forall x,y\in A, u,v\in V$)
$$(x+u)\succ (y+v)=x\succ y+l_\succ(x)v+r_\succ (y)u,\;\;
(x+u)\prec (y+v)=x\prec y+l_\prec(x)v+r_\prec(y)u.\eqno (3.1.10)$$
We denote it by $A\ltimes_{l_\succ,r_\succ, l_\prec,r_\prec}V$.


\begin{prop}{\rm (\cite{Bai3})}\quad Let
$(l_\succ,r_\succ,l_\prec,r_\prec$, $V)$ be a bimodule of a
dendriform algebra $(A,\succ,\prec)$. Let $(A,*)$ be the associated
associative algebra.

(1) Both $(l_\succ, r_\prec, V)$ and $(l_\succ+l_\prec,
r_\succ+r_\prec, V)$ are bimodules of $(A,*)$.

(2) For any bimodule $(l,r, V)$ of $(A,*)$, $(l,0,0,r, V)$ is a
bimodule of $(A,\succ,\prec)$.

(3) Both $(l_\succ+l_\prec, 0,0,r_\succ+r_\prec, V)$ and $(l_\succ,
0,0,r_\prec, V)$ are bimodules of $(A,\succ,\prec)$.

(4) The dendriform algebras
$A\ltimes_{l_\succ,r_\succ,l_\prec,r_\prec} V$ and
$A\ltimes_{l_\succ+l_\prec, 0,0,r_\succ+r_\prec}V$ have the same
associated associative algebra
$A\ltimes_{l_\succ+l_\prec,r_\succ+r_\prec} V$.\end{prop}

\begin{prop} {\rm (\cite{Bai3})}\quad Let
$(A,\succ,\prec)$ be a dendriform algebra and
$(l_\succ,r_\succ,l_\prec,r_\prec,V)$ be a bimodule. Then
$(r_\succ^*+r_\prec^*, -l_\prec^*,-r_\succ^*,l_\succ^*+l_\prec^*,
V^*)$ is a bimodule of $(A,\succ,\prec)$. We call it the dual
bimodule of $(l_\succ,r_\succ,l_\prec,r_\prec,V)$.\end{prop}

\begin{coro} {\rm (\cite{Bai3})}\quad  Let $(A,\succ,\prec)$ be a dendriform
algebra and $(l_\succ,r_\succ,l_\prec,r_\prec,V)$ be a bimodule. Let
$(A,*)$ be the associated associative algebra.

(1) Both $(r_\succ^*+r_\prec^*, 0,0,l_\succ^*+l_\prec^*, V^*)$ and
$(r_\prec^*, 0,0,l_\succ^*, V^*)$ are bimodules of
$(A,\succ,\prec)$.

(2) Both $(r_\succ^*+r_\prec^*, l_\succ^*+l_\prec^*,V^*)$ and
$(r_\prec^*, l_\succ^*, V^*)$ are bimodules of $(A,*)$.\end{coro}

\begin{exam}{\rm  Let $(A,\succ,\prec)$ be a dendriform algebra. Then
$$(L_\succ,R_\succ, L_\prec,R_\prec, A),\;\;\;(L_\succ,0,0,R_\prec, A)\;\;
{\rm and}\;\;(L_\succ+L_\prec,0,0,R_\succ+R_\prec, A)$$ are
bimodules of $(A,\prec, \succ)$ and the first one is called the
regular bimodule of $(A,\succ,\prec)$. On the other hand,
$$(R_\succ^*+R_\prec^*,-L_\prec^*, -R_\succ^*,L_\succ^*+L_\prec^*,
A^*),\;\;(R_\prec^*,0,0,L_\succ^*, A^*) \;\;{\rm
and}\;\;(R_\succ^*+R_\prec^*,0,0,L_\succ^*+L_\prec^*,A^*)$$ are
bimodules of $(A,\succ, \prec)$, too.}\end{exam}

\subsection{Bilinear forms on dendriform algebras and $D$-equation }

In fact, a double construction of Connes cocycle on an associative
algebra $(A,*)$ is equivalent to a double construction of certain
nondegenerate bilinear form on its compatible dendriform algebra.

\begin{defn}{\rm Let $(A,\succ,\prec)$ be a dendriform algebra. A
skew-symmetric bilinear form $\omega$ on A is called {\it invariant}
if $\omega$ satisfies (for any $x, y, z\in A$)
$$\omega(x\succ y,z)=\omega(y, z\succ x+z\prec x),\eqno (3.2.1)$$
$$\omega(x\prec y,z)=\omega(x,y\succ z+y\prec z).\eqno
(3.2.2)$$}\end{defn}

\begin{prop} Let $(A,\succ,\prec)$ be a dendriform algebra with
a skew-symmetric bilinear form $\omega$.

(1) $\omega$ is invariant if and only if $\omega$ satisfies equation
(3.2.1) and
$$\omega(x\succ y,z)+\omega (x, y\prec z)=0,\;\;\forall \; x,y,z\in
A.\eqno (3.2.3)$$

(2) $\omega$ is invariant if and only if $\omega$ satisfies
equations (3.2.2) and (3.2.3).

(3) $\omega$ is invariant if and only if $\omega$ satisfies equation
(3.2.1) and
$$\omega (x\succ y,z)+\omega (y\succ z,x)+\omega (z\succ x,
y)=0,\;\;\forall x,y,z\in A.\eqno (3.2.4)$$

(4) $\omega$ is invariant if and only if $\omega$ satisfies equation
(3.2.2) and
$$\omega (x\prec y,z)+\omega (y\prec z,x)+\omega (z\prec x,
y)=0,\;\;\forall x, y, z\in A. \eqno (3.2.5)$$

(5) If $\omega$ is invariant, then $\omega$ is a  Connes cocycle of
the associated associative algebra $(A,*)$.\end{prop}

\begin{proof} (1) Let $\omega$ be a skew-symmetric bilinear form satisfying
equation (3.2.1). If $\omega$ satisfies equation (3.2.2), then
$$\omega(x\succ y,z)=\omega (y, z\succ x+z\prec x)=\omega (y\prec
z,x)=-\omega (x,y\prec z),\;\;\forall x, y, z\in A.$$ Conversely, if
$\omega$ satisfies equation (3.2.3), then
$$\omega(x\prec y,z)=\omega (z\succ x, y)=\omega(x, y\succ z+y\prec
z),\;\;\forall x,y,z\in A.$$

 By a similar discussion as in (1), the conclusion (2) holds.

 (3) If the skew-symmetric bilinear form $\omega$ satisfies equation (3.2.1),
then
$$\omega(x\succ y,z)=\omega (y, z\succ x+z\prec x)=-\omega (z\succ
x,y)-\omega (z\prec x, y).\;\;\forall x,y,z\in A.$$ Therefore
equation (3.2.3) holds if and only if equation (3.2.4) holds. By
(1), the conclusion (3) follows.

The conclusion (4) follows by a similar discussion as in (3).

(5) follows immediately from the sum of equations (3.2.1) and
(3.2.2).
\end{proof}

By Corollary 2.2.7 and the conclusion (5) in Proposition 3.2.2, we
have the following result.

\begin{coro} Let $(A,*)$ be an associative algebra and $\omega$ be a
nondegenerate skew-symmetric bilinear form. Then $\omega$ is a
 Connes cocycle of $(A,*)$ if and only if
$\omega$ is invariant on the compatible dendriform algebra given by
equation (2.2.8).\end{coro}

Let $(A,\succ_A,\prec_A)$ be a dendriform algebra and suppose that
there is another dendriform algebra structure
$\succ_{A^*},\prec_{A^*}$ on its dual space $A^*$. If there is a
dendriform algebra structure on the direct sum $A\oplus A^*$ of
the underlying vector spaces of $A$ and $A^*$ such that both $A$
and $A^*$ are subalgebras and the skew-symmetric bilinear form on
$A\oplus A^*$ given by equation (2.2.10) is invariant on $A\oplus
A^*$, then it is called a {\it double construction of a dendriform
algebra with a nondegenerate invariant bilinear form associated to
$(A,\succ_A$,$\prec_A)$ and $(A^*,\succ_{A^*},\prec_{A^*})$}.

By Corollary 3.2.3 and Theorem 2.3.2, we have the following
conclusion.

\begin{coro} Let $(A,\succ_A,\prec_A)$ be a dendriform algebra whose
products are given by two linear maps $\beta_\succ^*,
\beta_\prec^* :A\otimes A\rightarrow A$. Suppose there is another
dendriform algebra structure $``\succ_{A^*}$, $\prec_{A^*}"$ on
its dual space $A^*$ given by two linear maps
$\alpha_\succ^*,\alpha_\prec^*: A^*\otimes A^*\rightarrow A^*$.
Then there exists a double construction of a dendriform algebra
with a nondegenerate invariant bilinear form associated to
$(A,\succ_A,\prec_A)$ and $(A^*,\succ_{A^*},\prec_{A^*})$ if and
only if
$(A,A^*,\alpha_\succ,\alpha_\prec,\beta_\succ,\beta_\prec)$ is a
dendriform D-bialgebra. Moreover, every double construction of a
dendriform algebra with a nondegenerate invariant bilinear form
can be obtained from the above way.\end{coro}

The dendriform algebra structure on $A\oplus A^*$ is given by (for
any $x,y\in A$ and $a^*,b^*\in A^*$)
$$(x+a^*)\succ
(y+b^*)=x\succ_Ay+ R^*_{A^*}(a^*)y-L^*_{\prec_{A^*}}(b^*)x
+a^*\succ_{A^*}b^*+R^*_A(x)b^*-L^*_{\prec_A}(y)a^*,\eqno (3.2.6)$$
$$(x+a^*)\prec
(y+b^*)=x\succ_Ay- R^*_{\succ_{A^*}}(a^*)y+L^*_{A^*}(b^*)x
+a^*\prec_{A^*}b^*-R^*_{\succ_A}(x)b^*+L^*_A(y)a^*,\eqno (3.2.7)$$
where $R^*_{A^*}=R^*_{\prec_{A^*}}+R^*_{\succ_{A^*}},
R^*_A=R^*_{\prec_A}+R^*_{\succ_A},
L^*_{A^*}=L^*_{\prec_{A^*}}+R^*_{\succ_{A^*}},
L^*_A=L^*_{\prec_A}+R^*_{\succ_A}$.

\begin{coro} Let $(A,\succ,\prec)$ be a dendriform algebra and $r\in
A\otimes A$. Let the linear maps $\alpha_\succ,\alpha_\prec$ be
defined by equations (2.3.8) and (2.3.9). If $r$ is a symmetric
solution of $D$-equation in $A$, then the maps
$\alpha_\succ,\alpha_\prec$ induce a dendriform algebra structure
on $A^*$ such that there is a double construction of a dendriform
algebra with a nondegenerate invariant bilinear form associated to
$(A,\succ_A,\prec_A)$ and
$(A^*,\succ_{A^*},\prec_{A^*})$.\end{coro}

\begin{remark}{\rm In the above sense, the $D$-equation in a dendriform
algebra is just an analogue of the classical Yang-Baxter equation in
a Lie algebra.}\end{remark}

Next, we consider the symmetric bilinear forms on a dendriform
algebra.

\begin{theorem}{\rm (\cite{Bai3})}\quad Let $(A,\succ,\prec)$ be a dendriform
algebra and $r\in A\otimes A$. Suppose that $r$ is symmetric and
nondegenerate. Then $r$ is a solution of $D$-equation in $A$ if and
only if the inverse of the isomorphism $A^*\rightarrow A$ induced by
$r$, regarded as a bilinear form ${\mathcal B}$ on $A$ (that is,
${\mathcal B}(x,y)=\langle r^{-1}x,y\rangle $ for any $x,y\in A$)
satisfies
$${\mathcal B}(x*y,z)={\mathcal B}(y,z\prec x)+{\mathcal B}(x,y\succ
z),\;\;\forall\; x,y,z\in A.\eqno (3.2.8)$$\end{theorem}

\begin{defn} {\rm Let $(A,\succ,\prec)$ be a dendriform algebra. A symmetric
bilinear form ${\mathcal B}$ on A is called {\it 2-cocycle} of
$(A,\succ,\prec)$ if ${\mathcal B}$  satisfies equation
(3.2.8).}\end{defn}

\subsection{ ${\mathcal O}$-operators of dendriform algebras and $D$-equation}

\begin{defn}{\rm  Let $(A,\succ,\prec)$ be a dendriform
algebra and $(l_\succ,r_\succ,l_\prec,r_\prec, V)$ be a bimodule.
A linear map $T:V\rightarrow A$ is called an {\it ${\mathcal
O}$-operator of $(A,\succ,\prec)$ associated to
$(l_\succ,r_\succ,l_\prec,r_\prec, V)$} if $T$ satisfies  (for any
$u,v\in V$)
$$T(u)\succ T(v)=T(l_\succ(T(u))v+r_\succ(T(v)u)),T(u)\prec T(v)=T(l_\prec(T(u))v+r_\prec(T(v)u)). \eqno (3.3.1)$$
}\end{defn}

The following result is obvious.

\begin{coro}Let $(l_\succ,r_\succ,l_\prec,r_\prec, V)$ be a  bimodule of
a dendriform algebra $(A,\succ,\prec)$. Let $(A,*)$ be the
associated associative algebra. If $T$ is an ${\mathcal O}$-operator
of $(A,\succ,\prec)$ associated to
$(l_\succ,r_\succ,l_\prec,r_\prec, V)$, then $T$ is an $\mathcal
O$-operator of $(A,*)$ associated to $(l_\succ+l_\prec,
r_\succ+r_\prec, V)$.\end{coro}

\begin{theorem}  Let $(A,\succ,\prec)$ be a dendriform algebra and $(A,*)$ be
the associated associative algebra. Let $r\in A\otimes A$ be
symmetric. Then the following conditions are equivalent.

(1) $r$ is a solution of $D$-equation in $(A,\succ,\prec)$.

(2) $r$ is an ${\mathcal O}$-operator of $(A,*)$ associated to
$(R_\prec^*, L_\succ^*)$.

(3) $r$ satisfies
$$r(a^*)\succ r(b^*)=r(R_*^*(r(a^*))b^*-L_\prec^*(r(b^*))a^*),\;\;\forall a^*,b^*\in A^*.\eqno (3.3.2)$$

(4) $r$ satisfies
$$r(a^*)\prec r(b^*)=r(-R_\succ^*(r(a^*))b^*+L_*^*(r(b^*))a^*),\;\;\forall a^*,b^*\in A^*.\eqno
(3.3.3)$$\end{theorem}

\begin{proof} The fact that (1) is equivalent to (2) follows from
Proposition 2.3.5. Let $\{ e_1,\cdots,e_n\}$ be a basis of $A$ and
$\{ e_1^*,\cdots,e_n^*\}$ be its dual basis. Suppose that
$$e_i\succ e_j=\sum_k a_{ij}^k e_k,\;\;e_i\prec e_j=\sum_k b_{ij}^k e_k,\;\;r=\sum_{i,j}r_{ij}e_i\otimes e_j,\;\;
r_{ij}=r_{ji}.$$ Hence $r(e_i^*)=\sum_k r_{ik}e_k$. Then $r$ is a
solution of $D$-equation in $A$ if and only if (for any $m,t,p$)
$$\sum_{i,k}[
r_{it}r_{kp}(a_{ik}^m+b_{ik}^m)-r_{mi}r_{tk}b_{ik}^p-r_{ip}r_{mk}a_{ik}^t]=0.$$
The left-hand side of the above equation is precisely the
coefficient of $e_t$ in
$$-r(e_p^*)\succ
r(e_m^*)+r(R_*^*(r(e_p^*))e_m^*-L_\prec^*(r(e_m^*))e_p^*),$$ and the
coefficient of $e_p$ in
$$-r(e_m^*)\prec r(e_t^*)+r(-R_\succ^*(r(e_m^*))e_t^*+L_*^*(r(e_t^*))e_m^*).$$
Therefore the conclusion follows.
\end{proof}

\begin{coro}  Let $(A,\succ,\prec)$ be a dendriform algebra and $r\in
A\otimes A$ be symmetric. Then $r$ is a solution of $D$-equation
in  $(A,\succ,\prec)$ if and only if $r$ is an ${\mathcal
O}$-operator of $(A,\succ,\prec)$ associated to
$(R_\succ^*+R_\prec^*,-L_\prec^*, -R_\succ^*,L_\succ^*+L_\prec^*,
A^*)$.\end{coro}

\begin{theorem} Let $(A,\succ,\prec)$ be a dendriform algebra. Let
$(l_\succ,r_\succ,l_\prec,r_\prec,V)$ be a bimodule and
$(r_\succ^*+r_\prec^*, -l_\prec^*,-r_\succ^*,l_\succ^*+l_\prec^*,
V^*)$ be the dual bimodule given by Proposition 3.1.3. Let
$T:V\rightarrow A$ be a linear map identified as an element in
$A\otimes V^*$ which is in the underlying vector space of $
(A\ltimes_{r_\succ^*+r_\prec^*,
-l_\prec^*,-r_\succ^*,l_\succ^*+l_\prec^*}V^*)\otimes
(A\ltimes_{r_\succ^*+r_\prec^*,
-l_\prec^*,-r_\succ^*,l_\succ^*+l_\prec^*}V^*)$. Then
$r=T+\sigma(T)$ is a symmetric solution of $D$-equation in the
dendriform algebra $A\ltimes_{r_\succ^*+r_\prec^*,
-l_\prec^*,-r_\succ^*,l_\succ^*+l_\prec^*}V^*$ if and only if $T$ is
an ${\mathcal O}$-operator of $(A,\succ,\prec)$ associated to
$(l_\succ,r_\succ,l_\prec,r_\prec,V)$.\end{theorem}

\begin{proof}
Let $\{e_1,\cdots,e_n\}$ be a basis of $A$. Let $\{v_1,\cdots,
v_m\}$ be a basis of $V$ and $\{ v_1^*,\cdots, v_m^*\}$ be its dual
basis. Set $T(v_i)=\sum\limits_{j=1}^na_{ij}e_j, i=1,\cdots, m$.
Since ${\rm Hom}(V,A)\cong A\otimes V^*$ as vector spaces,
\begin{eqnarray*}
T&=&\sum_{i=1}^m T(v_i)\otimes v_i^*=\sum_{i=1}^m\sum_{j=1}^n
a_{ij}e_j\otimes v_i^*\in A\otimes V^*\\&\subset&
(A\ltimes_{r_\succ^*+r_\prec^*,
-l_\prec^*,-r_\succ^*,l_\succ^*+l_\prec^*}V^*)\otimes
(A\ltimes_{r_\succ^*+r_\prec^*,
-l_\prec^*,-r_\succ^*,l_\succ^*+l_\prec^*}V^*).\end{eqnarray*}
Therefore we have
\begin{eqnarray*}  r_{12}* r_{13}
&=&\sum_{i,k=1}^m \{T(v_i)* T(v_k)\otimes v_i^*\otimes v_k^*
+r_\prec^*(T(v_i))v_k^*\otimes v_i^*\otimes T(v_k)\\
&\mbox{}&\hspace{1cm}+ l_\succ^* (T(v_k))v_i^*\otimes
T(v_i)\otimes v_k^* \};\\
r_{13}\prec r_{23} & =&\sum_{i,k=1}^m \{T(v_i)\otimes v_k^*\otimes
(l_\succ^*+l^*_\prec)(T(v_k))v_i^*- v_i^*\otimes T(v_k)\otimes
r_\succ^*(T(v_i))v_k^*\\
&\mbox{}&\hspace{1cm}
+v_i^*\otimes v_k^*\otimes T(v_i)\prec T(v_k)\};\\
\end{eqnarray*}
\begin{eqnarray*}r_{23}\succ r_{12}&=&\sum_{i,k=1}^m\{ T(v_k)\otimes
(r_\succ^*+r_\prec^*)(T(v_i))v_k^*\otimes v_i^*+v_k^*\otimes
T(v_i)\succ T(V_k)\otimes v_i^*\\
&\mbox{}&\hspace{1cm}-v_k^*\otimes l_\prec^*(T(v_j))v_i^*\otimes
T(v_i)\}.
\end{eqnarray*}
Furthermore, by equation (1.5), we show that
$$l_\succ^*(T(v_k))v_i^*=\sum_{j=1}^m v_i^*(l_\succ(T(v_k))v_j) v_j^*,\;\;\;
r_\succ^*(T(v_k))v_i^*=\sum_{j=1}^m v_i^*(r_\succ(T(v_k))v_j) v_j^*.
$$
$$l_\prec^*(T(v_k))v_i^*=\sum_{j=1}^m v_i^*(l_\prec(T(v_k))v_j) v_j^*,\;\;\;
r_\prec^*(T(v_k))v_i^*=\sum_{j=1}^m v_i^*(r_\prec(T(v_k))v_j) v_j^*.
$$
Thus
\begin{eqnarray*}
&&\sum_{i,k=1}^m T(v_i)\otimes
l_\succ^*(T(v_k))v_i^*=\sum_{i,k=1}^mT(v_i)\otimes
[\sum_{j=1}^m v_i^*(l_\succ(T(v_k))v_j) v_j^*]\\
&&=\sum_{i,k=1}^m \sum_{j=1}^m v_j^*(l_\succ(T(v_k))v_i)
T(v_j)\otimes v_i^* =\sum_{i,k=1}^m T(l_\succ(T(v_k))v_i)\otimes
v_i^*.
\end{eqnarray*}
Similarly, we have
\begin{eqnarray*}
\sum_{i,k=1}^m T(v_i)\otimes r_\succ^*(T(v_k))v_i^*&=&\sum_{i,k=1}^m
T(r_\succ(T(v_k))v_i)\otimes v_i^*;\\
\sum_{i,k=1}^m T(v_i)\otimes
l_\prec^*(T(v_k))v_i^*&=&\sum_{i,k=1}^m
T(l_\prec(T(v_k))v_i)\otimes v_i^*;\\
\sum_{i,k=1}^m T(v_i)\otimes r_\prec^*(T(v_k))v_i^*&=&\sum_{i,k=1}^m
T(r_\prec(T(v_k))v_i)\otimes v_i^*.
\end{eqnarray*}
Therefore
\begin{eqnarray*}
&&r_{12}*r_{13}-r_{13}\prec r_{23}- r_{23}\succ r_{12}\\
&&=\sum_{i,k=1}^m\{ (T(v_i)*
T(v_k)-T((r_\succ+r_\prec)(T(v_k))v_i)-T((l_\succ+l_\prec)(T(v_i))v_k))\otimes
v_i^*\otimes
v_k^*\\
&& +v_i^*\otimes (-T(v_i)\succ
T(v_k)+T(r_\succ(T(v_k))v_i)+T(l_\succ(T(v_i))v_k))\otimes
v_k^*\\
&& + v_i^*\otimes v_k^*\otimes (-T(v_i)\prec
T(v_k)+T(r_\prec(T(v_k))v_i)+T(l_\prec(T(v_i))v_k))\}.
\end{eqnarray*}
So $r$ is a symmetric solution of $D$-equation in the dendriform
algebra $A\ltimes_{r_\succ^*+r_\prec^*,
-l_\prec^*,-r_\succ^*,l_\succ^*+l_\prec^*}V^*$ if and only if $T$ is
an ${\mathcal O}$-operator of $(A,\succ,\prec)$ associated to
$(l_\succ,r_\succ,l_\prec,r_\prec,V)$.
\end{proof}

\begin{coro} Let $(A,*)$ be an associative algebra. Let $(l,r,V)$ be
a bimodule and $(r^*,l^*,V^*)$ be the dual bimodule given in Example
2.1.2. Suppose that $T:V\rightarrow A$ is an ${\mathcal O}$-operator
of $(A,*)$ associated to $(l,r,V)$. Then $r=T+\sigma(T)$ is a
symmetric solution of $D$-equation in the dendriform algebra
$T(V)\ltimes_{r^*,0,0,l^*}V^*$, where $T(V)\subset A$ is a
dendriform algebra given by equation (2.1.8) and $(r^*,0,0,l^*)$ is
its bimodule since its associated associative algebra $T(V)$ is an
associative subalgebra of $A$, and $T$ can be identified as an
element in $T(V)\otimes V^*$ which is in the underlying vector space
of $(T(V)\ltimes_{r^*,0,0,l^*}V^*)\otimes
(T(V)\ltimes_{r^*,0,0,l^*}V^*)$.\end{coro}

\begin{proof}
By Theorem 2.1.8, it is obvious that $T:V\rightarrow T(V)$ is an
${\mathcal O}$-operator of $(T(V),\succ,\prec)$ associated to the
bimodule $(l,0,0,r,V)$, where
$$T(u)\succ T(v)=T(l(T(u))v),\;\;T(u)\prec
T(v)=T(r(T(v))u),\;\;\forall\;u,v\in V.$$ Hence the conclusion
follows from Theorem 3.3.5 immediately.
\end{proof}

\begin{remark}{\rm The above conclusion has appeared in \cite{Bai3} with a direct
proof. We would like to emphasis that it involves only the
${\mathcal O}$-operators of associative algebras (not the ${\mathcal
O}$-operators of dendriform algebras). Therefore, as has been
pointed out in \cite{Bai3}, roughly speaking, the symmetric part of
an ${\mathcal O}$-operator of an associative algebra corresponds to
a symmetric solution of $D$-equation, whereas the skew-symmetric
part of an ${\mathcal O}$-operator of an associative algebra
corresponds to a skew-symmetric solution of associative Yang-Baxter
equation.}\end{remark}

\begin{coro}{\rm (cf. Proposition 3.4.12)}\quad Let $(A,\succ, \prec)$ be a
dendriform algebra. Then
$$r=\sum_{i}^n (e_i\otimes e_i^*+e_i^*\otimes e_i)\eqno (3.3.4)$$
is a symmetric solution of $D$-equation in the dendriform algebra
$A\ltimes_{R_\prec^*, 0,0, L_\succ^*}A^*$, where $\{e_1,\cdots$, $
e_n\}$ is a basis of $A$ and $\{e_1^*,\cdots, e_n^*\}$ is its dual
basis. Moreover there is a natural 2-cocycle $\mathcal B$ of the
dendriform algebra $A\ltimes_{R_\prec^*, 0,0, L_\succ^*}A^*$ induced
by $r^{-1}: A\oplus A^*\rightarrow (A\oplus A^*)^*$, which is given
by equation (2.2.11).\end{coro}

\begin{proof}
Since ${id}$ is an ${\mathcal O}$-operator of $(A,\succ,\prec)$
associated to the bimodule $(L_\succ,0,0, R_\prec,A)$, $r$ is a
symmetric solution of $D$-equation in $A\ltimes_{R_\prec^*, 0,0,
L_\succ^*}A^*$ due to Theorem 3.3.5. Therefore the bilinear form
$\mathcal B$ given by equation (2.2.11) is a 2-cocycle due to
Theorem 3.2.7.
\end{proof}

\subsection{${\mathcal O}$-operators of dendriform algebras and quadri-algebras}

\begin{defn}{\rm (\cite{AL})\quad Let $A$ be a vector space
 with four bilinear products denoted by
$\searrow, \nearrow, \nwarrow$ and $\swarrow : A\otimes A\rightarrow
A$. $(A, \searrow, \nearrow, \nwarrow, \swarrow)$ is called a {\it
quadri-algebra} if for any $x,y,z\in A$,
$$(x\nwarrow y)\nwarrow z=x\nwarrow (y* z),\;\; (x\nearrow y)\nwarrow
z=x\nearrow (y\prec z),\;\;(x\wedge y)\nearrow z=x\nearrow (y\succ
z),\eqno (3.4.1)$$
$$(x\swarrow y)\nwarrow z=x\swarrow (y\wedge z),\;\; (x\searrow y)\nwarrow
z=x\searrow (y\nwarrow z),\;\;(x\vee y)\nearrow z=x\searrow
(y\nearrow z),\eqno (3.4.2)$$
$$(x\prec y)\swarrow z=x\swarrow (y\vee z),\;\; (x\succ y)\swarrow
z=x\searrow (y\swarrow z),\;\;(x* y)\searrow z=x\searrow (y\searrow
z),\eqno(3.4.3)$$ where
$$x\succ y=x\nearrow y+x\searrow y,
x\prec y=x\nwarrow y+x\swarrow y, x\vee y=x\searrow y+x\swarrow
y,x\wedge y=x\nearrow y+x\nwarrow y,\eqno(3.4.4)$$ and
$$x*y=x\searrow y+x\nearrow y+x\nwarrow y+x\swarrow y=x\succ
y+x\prec  y=x\vee y+x\wedge y.\eqno (3.4.5)$$}\end{defn}

\begin{prop}{\rm (\cite{AL})}\quad Let $(A, \searrow, \nearrow, \nwarrow,
\swarrow)$ be a quadri-algebra.

(1) The product given by
$$x\succ y=x\nearrow y+x\searrow y,\;\;
x\prec y=x\nwarrow y+x\swarrow y,\;\;\forall x,y\in A,\eqno
(3.4.6)$$ defines a dendriform algebra. $(A,\succ,\prec)$ is
called the associated horizontal dendriform algebra of $(A,
\searrow, \nearrow, \nwarrow, \swarrow)$ and $(A, \searrow,
\nearrow, \nwarrow, \swarrow)$ is called a compatible
quadri-algebra structure on the horizontal dendriform algebra
$(A,\succ,\prec)$.

(2) The product given by
$$x\vee y=x\searrow y+x\swarrow
y,\;x\wedge y=x\nearrow y+x\nwarrow y,\;\;\forall x,y\in A,\eqno
(3.4.7)$$ defines a dendriform algebra.  $(A,\vee,\wedge)$ is called
the associated vertical dendriform algebra of $(A, \searrow,
\nearrow, \nwarrow, \swarrow)$ and $(A, \searrow, \nearrow,
\nwarrow, \swarrow)$ is called a compatible quadri-algebra structure
on the (vertical) dendriform algebra $(A,\vee,\wedge)$.

(3) The product given by equation (3.4.5) defines an associative
algebra. $(A,*)$ is called the associated associative algebra of
$(A, \searrow, \nearrow, \nwarrow, \swarrow)$ and $(A, \searrow,
\nearrow, \nwarrow, \swarrow)$ is called a compatible quadri-algebra
structure on the associative algebra $(A,*)$.\end{prop}

\begin{prop} Let $A$ be a vector space  with four
bilinear products denoted by $\searrow, \nearrow, \nwarrow$ and
$\swarrow : A\otimes A\rightarrow A$.

(1) $(A, \searrow, \nearrow, \nwarrow, \swarrow)$ is a
quadri-algebra if and only if $(A,\succ,\prec)$ defined by
equation (3.4.6) is a dendriform algebra and $(L_\searrow,
R_\nearrow, L_\swarrow, R_\nwarrow, A)$ is a bimodule.

(2) $(A, \searrow, \nearrow, \nwarrow, \swarrow)$ is a
quadri-algebra if and only if $(A,\vee,\wedge)$ defined by
equation (3.4.7) is a dendriform algebra and $(L_\searrow,
R_\swarrow, L_\nearrow,R_\nwarrow, A)$ is a bimodule.\end{prop}

\begin{proof}
The conclusions can be obtained by a direct computation or a similar
proof as of Proposition 3.4.6.
\end{proof}

\begin{coro} Let $(A, \searrow, \nearrow, \nwarrow, \swarrow)$ be a
quadri-algebra. Then $(L_\searrow, R_\nwarrow, A)$ is a bimodule of
the associated associative algebra $(A,*)$.\end{coro}

\begin{proof}
It follows immediately from Propositions 3.1.2 and 3.4.3.
\end{proof}

For brevity, we pay our main attention to the study of associated
horizontal dendriform algebras. In fact, the corresponding study on
the associated vertical dendriform algebras are completely similar.

\begin{coro} Let $(A, \searrow, \nearrow, \nwarrow, \swarrow)$ be a
quadri-algebra and $(A,\succ,\prec)$ be the associated horizontal
dendriform algebra. Then
$$(L_\succ,R_\succ, L_\prec,R_\prec, A),\;\;\;(L_\succ,0,0,R_\prec, A),\;\;\;
(L_\succ+L_\prec,0,0,R_\succ+R_\prec, A);$$
$$(L_\searrow, R_\nearrow, L_\swarrow,
R_\nwarrow, A),\;\;\;(L_\searrow,0,0, R_\nwarrow, A)\;\;{\rm
and}\;\;(L_\vee,0,0,R_\wedge,A)$$ are bimodules of $(A,\prec,
\succ)$. On the other hand,
$$(R_\succ^*+R_\prec^*,-L_\prec^*, -R_\succ^*,L_\succ^*+L_\prec^*,
A^*),\;\;(R_\prec^*,0,0,L_\succ^*, A^*)
\;\;(R_\succ^*+R_\prec^*,0,0,L_\succ^*+L_\prec^*,A^*);$$
$$(R_\nearrow^*+R_\nwarrow^*,-L_\swarrow^*,-R_\nearrow^*,
L_\searrow^*+L_\swarrow^*, A^*),\;\;(R_\nwarrow^*,0,0,L_\searrow^*,
A^*) \;\;{\rm and}\;\;(R_\wedge^*,0,0,L_\vee^*,A^*)$$ are bimodules
of $(A,\succ, \prec)$, too.\end{coro}

\begin{proof}
It follows immediately from Propositions 3.1.2, 3.1.3 and 3.4.3 and
Example 3.1.5.
\end{proof}

\begin{prop} Let $(A,\succ,\prec)$ be a dendriform algebra and
$(l_\succ,r_\succ,l_\prec,r_\prec, V)$ be a bimodule. Let
$T:V\rightarrow A$ be an ${\mathcal O}$-operator associated to
$(l_\succ,r_\succ,l_\prec,r_\prec, V)$. Then there exists a
quadri-algebra structure on $V$ given by (for any $u,v\in V$)
$$u\searrow v=l_\succ(T(u))v,\;u\nearrow v=r_\succ(T(v))u,\;
u\swarrow v=l_\prec(T(u))v,\;u\nwarrow v=r_\prec(T(v))u.\eqno
(3.4.8)$$ Therefore, there exists a dendriform algebra structure on
$V$ given by equation (3.4.6) and $T$ is a homomorphism of
dendriform algebras. Furthermore, $T(V)=\{T(v)|v\in V\}\subset A$ is
a dendriform subalgebra of $A$ and there is an induced
quadri-algebra structure on $T(V)$ given by
$$T(u)\searrow T(v)=T(u\searrow v),\;T(u)\nearrow T(v)=T(u\nearrow v),$$
$$T(u)\swarrow T(v)=T(u\swarrow v),\;T(u)\nwarrow T(v)=T(u\nwarrow
v),\;\forall u,v\in V.\eqno (3.4.9)$$ Moreover, its corresponding
associated horizontal dendriform algebra structure on $T(V)$ given
by equation (3.4.6) is just the dendriform subalgebra structure of
$(A,\succ,\prec)$ and $T$ is a homomorphism of quadri-algebras.
\end{prop}

\begin{proof}Set $l=l_\succ+l_\prec$ and
$r=r_\succ+r_\prec$. For any $u,v,w\in V$, we have
\begin{eqnarray*}
(u\nwarrow v)\nwarrow w &=& r_\prec (T(w))(u\nwarrow
v)=r_\prec(T(w))r_\prec(T(v))u\stackrel{(3.1.3)}{=}
r_\prec(T(v)*T(w))u\\
&=& u\nwarrow (l(T(v))w+r(T(w))v)=u\nwarrow (v* w);\\
(u\nearrow v)\nwarrow w&=& r_\prec (T(w))(u\nwarrow v)
=r_\prec(T(w))r_\succ(T(v))u\stackrel{(3.1.6)}{=}
r_\succ(T(v)\prec T(w))u\\
&=& u\nwarrow (l_\prec(T(v))w+r_\prec(T(w))v)=u\nwarrow (v\prec w);\\
(u\wedge v)\nearrow w&=& r_\succ (T(w))(u\nearrow v+u\nwarrow v)
=r_\succ(T(w))(r_\succ(T(v))u+r_\prec(T(v))u)\\
&\stackrel{(3.1.9)}{=}& r_\succ(T(w)\succ T(v))u
= u\nearrow (l_\succ(T(w))v+r_\succ(T(v))w)=u\nearrow (v\succ w);\\
(u\swarrow v)\nwarrow w &=& r_\prec (T(w))(u\swarrow
v)=r_\prec(T(w))l_\prec(T(u))v\\
&\stackrel{(3.1.2)}{=}& l_\prec(T(u))(r_\prec(T(w))+r_\succ(T(w)))v=
u\swarrow (v\wedge w);\\
(u\searrow v)\nwarrow w&=&r_\prec(T(w))(u\searrow
v)=r_\prec(T(w))l_\succ(T(u))v\stackrel{(3.1.5)}{=}l_\succ(T(u))r_\succ(T(w)))v\\
&=&l_\succ(T(u))(v\nwarrow w)=u\searrow (v\nwarrow w);\\
(u\vee v)\nearrow
w&=&r_\succ(T(w))(l_\succ(T(u))+l_\prec(T(u))v\stackrel{(3.1.8)}{=}l_\succ(T(u))r_\succ(T(w))v\\
&=&l_\succ(T(u))(v\nearrow w)=u\searrow (v\nearrow w);\\
(u\prec v)\swarrow w&=& l_\succ(T(u\prec v))w=l_\prec(T(u)\prec
T(v))w\\
&\stackrel{(3.1.1)}{=}&l_\prec(T(u))(l_\prec(T(v))+l_\succ(T(v)))w
=u\swarrow (v\vee w); \\
\end{eqnarray*}
\begin{eqnarray*}
(u\succ v)\swarrow w&=&l_\prec(T(u\succ v))w=l_\prec(T(u)\succ T(v))
w\stackrel{(3.1.4)}{=}
l_\succ(T(u))l_\prec(T(v))w\\
&=&l_\succ(T(u))(v\swarrow w)=u\searrow (u\swarrow w);\\
(u* v)\searrow w &=& l_\succ(T(u*
v))w=l_\succ(T(u)*T(v))w\stackrel{(3.1.7)}{=}l_\succ(T(u))l_\succ(T(v))w\\
&=&l_\succ(T(u))(v\searrow w)=u\searrow (v\searrow w).
\end{eqnarray*}
Therefore $(V, \searrow, \nearrow, \nwarrow, \swarrow)$ is a
quadri-algebra. Furthermore the other results follow easily.
\end{proof}

\begin{defn}{\rm (\cite{AL})\quad Let $(A,\succ,\prec)$ be a dendriform
algebra. An $\mathcal O$-operator $R$ of $(A,\succ$,$\prec)$
associated to the regular bimodule $(L_\succ, R_\succ, L_\prec,
R_\prec, A)$ is called a {\it Rota-Baxter operator} on $A$, that is,
$R$ satisfies
$$R(x\succ y)=R(R(x)\succ y+x\succ R(y)),\;
R(x\prec y)=R(R(x)\prec y+x\prec R(y)),\;\forall x,y\in A.\eqno
(3.4.10)$$}\end{defn}

By Proposition 3.4.6, the following conclusion follows immediately.

\begin{coro} {\rm (\cite{AL}, Proposition 2.3)}\quad Let $(A,\succ,\prec)$ be a
dendriform algebra and $R$ be a Rota-Baxter operator on
$(A,\succ,\prec)$. Then there exists a quadri-algebra structure on
$A$ defined by (for any $x,y\in A$)
$$x\searrow y=R(x)\succ y,\;x\nearrow y=x\succ R(y),\;
x\swarrow y=R(x)\prec y,\;x\nwarrow y=x\prec R(y).\eqno
(3.4.11)$$\end{coro}

Moreover, by Corollaries 2.1.9 and 3.4.8, it is obvious that the
Rota-Baxter operators on associative algebras can construct
quadri-algebras as follows (also see Lemma 4.4.9).

\begin{coro}{\rm (\cite{AL}, Corollary 2.6)}\quad Let $R_1$ and $R_2$ be a pair
of commuting Rota-Baxter operators (of weight zero) on an
associative algebra $(A,*)$. Then there exists a quadri-algebra
structure on $A$ defined by (for any $x,y\in A$)
$$x\searrow y=R_1R_2(x)* y,\;x\nearrow y=R_1(x)*R_2(y),\;
x\swarrow y=R_2(x)* R_1(y),\;x\nwarrow y=x*R_1R_2(y).\eqno
(3.4.12)$$\end{coro}


\begin{coro} Let $(A,\succ,\prec)$ be a dendriform algebra. Then there
exists a compatible quadri-algebra structure on $(A,\succ,\prec)$
such that $(A,\succ,\prec)$ is the associated horizontal
dendriform algebra if and only if there exists an invertible
$\mathcal O$-operator of $(A,\succ,\prec)$.
\end{coro}

\begin{proof} If there exists an invertible $\mathcal O$-operator
$T$ of $(A,\succ,\prec)$ associated to a bimodule $(l_\succ$,
$r_\succ$, $l_\prec,r_\prec, V$), then by Proposition 3.4.6, there
exists a quadri-algebra structure on $V$ given by equation (3.4.8).
Therefore we can define a quadri-algebra structure on $A$ by
equation (3.4.9) such that $T$ is an isomorphism of quadri-algebras,
that is,
$$x\searrow y=T(l_\succ(x)T^{-1}(y)),\;
x\nearrow y=T(r_\succ(y)T^{-1}(x)),\;$$$$ x\swarrow
y=T(l_\prec(x)T^{-1}(y)),\; x\nwarrow
y=T(r_\prec(y)T^{-1}(x)),\;\;\forall x,y\in A.$$ Moreover it is a
compatible quadri-algebra structure on $(A,\succ,\prec)$ since for
any $x,y\in A$, we have
$$x\succ y=T(T^{-1}(x)\succ T^{-1}(y))=T(r_\succ(y)T^{-1}(x)+l_\succ(x)T^{-1}(y))
=x\nearrow y+x\searrow y,$$
$$x\prec y=T(T^{-1}(x)\prec T^{-1}(y))=T(r_\prec(y)T^{-1}(x)+l_\prec(x)T^{-1}(y))
=x\nwarrow y+x\swarrow y.$$

Conversely, let $(A, \searrow, \nearrow, \nwarrow, \swarrow)$ be a
quadri-algebra and $(A,\succ,\prec)$ be the associated horizontal
dendriform algebra. Then $(L_\searrow, R_\nearrow, L_\swarrow,
R_\nwarrow, A)$ is a bimodule of $(A,\succ, \prec)$ and the
identity map $id$ is an $\mathcal O$-operator of $(A,\succ,\prec)$
associated to it.
\end{proof}

\begin{prop} Let $(A,\succ,\prec)$ be a dendriform algebra and $(A,*)$ be
the associated associative algebra. If there is a nondegenerate
symmetric 2-cocycle ${\mathcal B}$ of $(A,\succ,\prec)$, then
there exists a compatible quadri-algebra structure on
$(A,\succ,\prec)$ defined by (for any $x,y,z\in A$)
$${\mathcal B}(x\searrow y,z)={\mathcal B}(y,z* x)=\mathcal
B(y,z\succ x+z\prec x),\eqno (3.4.13)$$
$${\mathcal B}(x\nearrow y, z)=-{\mathcal B}(x,y\prec
z),\eqno (3.4.14)$$$${\mathcal B}(x\nwarrow y)={\mathcal B}(x,y*
z)=\mathcal B(x,y\succ z+y\prec z),\eqno (3.4.15)$$
$${\mathcal B}(x\swarrow y,z)=-{\mathcal B}(y,z\succ x).\eqno (3.4.16)$$
 such that $(A,\succ,\prec)$ is the associated horizontal
dendriform algebra.
\end{prop}

\begin{proof}
Since ${\mathcal B}$ is nondegenerate and symmetric, we define an
invertible linear map $T:A\rightarrow A^*$ by
$$\langle x, T(y)\rangle=\langle T(x),y\rangle=\mathcal B (x,y),\;\;\forall x,y\in A.$$
Let $x,y,z\in A$. Since ${\mathcal B}$ is a 2-cocycle of
$(A,\succ,\prec)$, we have
\begin{eqnarray*}
\langle T(x\succ y), z\rangle&=&{\mathcal B}(x\succ y,z)={\mathcal
B}(y,z*x)-{\mathcal B}(x,y\succ z)\\
&=&\langle T(y), z*x\rangle-\langle T(x), y\prec z\rangle =\langle
R^*_*(x)T(y),z\rangle-\langle L_\prec^*(y)T(x),z\rangle.
\end{eqnarray*}
So $$T(x\succ y)=R_*^*(x)T(y)-L_\prec^*(y)T(x),\;\;\forall x,y\in
A.$$ Similarly, we show that $$T(x\prec
y)=L_*^*(y)T(x)-R_\succ^*(x)T(y),\;\;\forall x,y\in A.$$ Hence
$T^{-1}$ is an (invertible) $\mathcal O$-operator of the dendriform
algebra $(A,\succ,\prec)$ associated to the bimodule $(R^*_*,
-L_\prec^*$, $ -R_\succ^*,L^*_*,A^*)$. Then by Corollary 3.4.10,
there exists a compatible quadri-algebra structure on
$(A,\succ,\prec)$ defined by
\begin{eqnarray*}
{\mathcal B}(x\searrow y,z)&=&\langle T(x\searrow y), z\rangle
=\langle T(T^{-1}(R_*^*(x)T(y))), z\rangle=\langle T(y),
z*x\rangle={\mathcal B}(y, z*x);\\
{\mathcal B}(x\nearrow y,z)&=&\langle T(x\nearrow y), z\rangle
=\langle T(T^{-1}(-L_\prec^*(y)T(x))), z\rangle=\langle T(x),
-y\prec z\rangle=-{\mathcal B}(x, y\prec z);\\
{\mathcal B}(x\nwarrow y,z)&=&\langle T(x\nwarrow y), z\rangle
=\langle T(T^{-1}(L_*^*(y)T(x))), z\rangle=\langle T(x),
y*z\rangle={\mathcal B}(x, y*z);\\
{\mathcal B}(x\swarrow y,z)&=&\langle T(x\swarrow y), z\rangle
=\langle T(T^{-1}(-R_\succ^*(x)T(y))), z\rangle=\langle T(y),
-z\succ y\rangle=-{\mathcal B}(y, z\succ x),
\end{eqnarray*}
such that $(A,\succ,\prec)$ is the associated horizontal dendriform
algebra.
\end{proof}

\begin{prop}{\rm (cf. Corollary 3.3.8)}\quad Let $(A, \searrow, \nearrow,
\nwarrow, \swarrow)$ be a quadri-algebra and $(A,\succ$, $\prec)$
be the associated horizontal dendriform algebra.
 Then $r$ given by equation (3.3.4)
is a symmetric solution of $D$-equation in the dendriform algebra
$A\ltimes_{R_\nearrow^*+R_\nwarrow^*,-L_\swarrow^*,-R_\nearrow^*,
L_\searrow^*+L_\swarrow^*}A^*$, where $\{e_1,\cdots, e_n\}$ is a
basis of $A$ and $\{e_1^*,\cdots, e_n^*\}$ is its dual basis.
Moreover there is a natural 2-cocycle $\mathcal B$  of the
dendriform algebra
$A\ltimes_{R_\nearrow^*+R_\nwarrow^*,-L_\swarrow^*,-R_\nearrow^*,
L_\searrow^*+L_\swarrow^*}A^*$ induced by $r^{-1}: A\oplus
A^*\rightarrow (A\oplus A^*)^*$, which is given by equation
(2.2.11).\end{prop}

\begin{proof}
By Corollary 3.4.5,
$(R_\nearrow^*+R_\nwarrow^*,-L_\swarrow^*,-R_\nearrow^*,
L_\searrow^*+L_\swarrow^*,A^*)$ is the dual bimodule of the
bimodule $(L_\searrow, R_\nearrow, L_\swarrow, R_\nwarrow, A)$ of
the associated horizontal dendriform algebra $(A,\succ,\prec)$.
Then the first half of the conclusion follows immediately from
Theorem 3.3.5 and the fact that ${id}$ is an ${\mathcal
O}$-operator of $(A,\succ,\prec)$ associated to the bimodule
$(L_\searrow, R_\nearrow, L_\swarrow, R_\nwarrow, A)$. The second
half of the conclusion follows from Theorem 3.2.7.
\end{proof}

At the end of this subsection, we give an algebraic equation on a
quadri-algebra as follows.

\begin{prop} Let $(A, \searrow, \nearrow, \nwarrow, \swarrow)$ be a
quadri-algebra and $(A,\succ,\prec)$ be the associated horizontal
dendriform algebra. Let $r\in A\otimes A$ be skew-symmetric. Then
$r$ is an ${\mathcal O}$-operator of $(A,\succ,\prec)$ associated
to the bimodule
$(R_\nearrow^*+R_\nwarrow^*,-L_\swarrow^*,-R_\nearrow^*,
L_\searrow^*+L_\swarrow^*,A^*)$ if and only if $r$ satisfies
$$r_{13}\succ r_{23}=r_{23}\nearrow r_{12}+r_{23}\nwarrow
r_{12}+r_{12}\swarrow r_{13},\eqno (3.4.17)$$ $$r_{13}\prec
r_{23}=-r_{23}\nearrow r_{12}-r_{12}\searrow r_{13}-r_{12}\swarrow
r_{13}.\eqno (3.4.18)$$\end{prop}

\begin{proof}
Let $\{ e_1,\cdots,e_n\}$ be basis of $A$ and
$\{e_1^*,\cdots,e_n^*\}$ be its dual basis. Suppose that
$$e_i\searrow e_j=\sum_k a_{ij}^k e_k,\;e_i\nearrow e_j=\sum_k b_{ij}^k e_k,\;\;
e_i\nwarrow e_j=\sum_k c_{ij}^k e_k,\;e_i\swarrow e_j=\sum_k
d_{ij}^k e_k,$$ and $r=\sum\limits_{i,j}r_{ij}e_i\otimes e_j$,
$r_{ij}=-r_{ji}$. Hence $r(e_i^*)=\sum\limits_k r_{ik}e_k$. Then $r$
satisfies equation (3.4.17) if and only if (for any $m,t,p$)
$$\sum_{i,k}[
r_{mi}r_{tk}(a_{ik}^p+b_{ik}^p)-r_{ip}r_{mk}(b_{ik}^t+c_{ik}^p)-r_{it}r_{kp}d_{ik}^m]=0.$$
The left-hand side of the above equation is precisely the
coefficient of $e_p$ in
$$r(e_m^*)\succ
r(e_t^*)-r((R_\nearrow^*+R_\nwarrow^*)(r(e_m^*))e_t^*-L_\swarrow^*(r(e_t^*))e_m^*).$$
On the other hand, $r$ satisfies equation (3.4.18) if and only if
(for any $m,t,p$)
$$\sum_{i,k}[
r_{mi}r_{tk}(c_{ik}^p+d_{ik}^p)+r_{ip}r_{mk}b_{ik}^t+r_{it}r_{kp}(d_{ik}^m+a_{ik}^m)]=0.$$
The left-hand side of the above equation is precisely the
coefficient of $e_p$ in
$$r(e_m^*)\prec
r(e_t^*)-r(-R_\nearrow^*(r(e_m^*))e_t^*+(L_\searrow^*+L_\swarrow^*)(r(e_t^*))e_m^*).$$
Therefore the conclusion follows.
\end{proof}

\begin{coro} Let $(A, \searrow, \nearrow, \nwarrow, \swarrow)$ be a
quadri-algebra and $(A,\succ,\prec)$ be the associated horizontal
dendriform algebra. Let $r\in A\otimes A$. If $r$ satisfies
equations (3.4.17) and (3.4.18), then $r$ satisfies
$$r_{13}*r_{23}=r_{23}\nwarrow r_{12}-r_{12}\searrow r_{13}.\eqno
(3.4.19)$$\end{coro}

\begin{proof}
In fact, equation (3.4.19) is the sum of equations (3.4.17) and
(3.4.18).
\end{proof}

\section{${\mathcal O}$-operators of quadri-algebras}

\subsection{Bimodule of quadri-algebras}

\begin{defn}{\rm Let $(A, \searrow, \nearrow, \nwarrow,
\swarrow)$ be a quadri-algebra and $V$ be a vector space. Let
$l_\searrow, r_\searrow, l_\nearrow, r_\nearrow,
l_\nwarrow,r_\nwarrow, l_\swarrow, r_\swarrow:A\rightarrow gl(V)$
be eight linear maps. $V$ (or $(l_\searrow, r_\searrow,
l_\nearrow, r_\nearrow, l_\nwarrow,r_\nwarrow,$ $l_\swarrow,
r_\swarrow)$, or $(l_\searrow, r_\searrow, l_\nearrow, r_\nearrow,
l_\nwarrow,r_\nwarrow, l_\swarrow, r_\swarrow,V)$) is called a
{\it bimodule} of $A$ if the following equations hold (for any
$x,y\in A$):
$$l_\nwarrow (x\nwarrow y)=l_\nwarrow
(x)l_*(y);\;\;r_\nwarrow(y)l_\nwarrow(x)=l_\nwarrow(x)r_*(y);\;\;
r_\nwarrow(y)r_\nwarrow(x)=r_\nwarrow(x* y);\eqno(4.1.1)$$
$$l_\nwarrow(x\swarrow y)=l_\swarrow(x)l_\wedge (y);\;\;
r_\nwarrow(y)l_\swarrow(x)=l_\swarrow(x)r_\wedge(y);\;\;
r_\nwarrow(y)r_\swarrow(x)=r_\swarrow(x\wedge y);\eqno(4.1.2)$$
$$l_\swarrow(x\prec y)=l_\swarrow(x)l_\vee (y);\;\;
r_\swarrow
(y)l_\prec(x)=l_\swarrow(x)r_\vee(y);\;\;r_\swarrow(y)r_\prec(x)=r_\swarrow(x\vee
y);\eqno(4.1.3)$$
$$l_\nwarrow(x\nearrow y)=l_\nearrow(x)l_\prec(y);\;\;
r_\nwarrow(y)l_\nearrow(x)=l_\nearrow(x)r_\prec(y);\;\;
r_\nwarrow(y)r_\nearrow(x)=r_\nearrow(x\prec y);\eqno (4.1.4)$$
$$l_\nwarrow(x\searrow y)=l_\searrow(x)l_\nwarrow(y);\;\;
r_\nwarrow(y)l_\searrow(x)=l_\searrow(x)r_\nwarrow(y);\;\;
r_\nwarrow(y)r_\searrow(x)=r_\searrow(x\nwarrow y);\eqno (4.1.5)$$
$$l_\swarrow(x\succ y)=l_\searrow(x)l_\swarrow(y);\;\;
r_\swarrow(y)l_\succ (x)=l_\searrow(x)r_\swarrow(y);\;\;
r_\swarrow(y)r_\succ(x)=r_\searrow(x\swarrow y);\eqno (4.1.6)$$
$$l_\nearrow(x\wedge y)=l_\nearrow(x)l_\succ(y);\;\;
r_\nearrow(y)l_\wedge (x)=l_\nearrow(x)r_\succ(y);\;
r_\nearrow(y)r_\wedge(x)=r_\nearrow(x\succ y);\eqno (4.1.7)$$
$$l_\nearrow(x\vee y)=l_\searrow(x)l_\nearrow(y);\;\;
r_\nearrow(y)l_\vee (x)=l_\searrow(x)r_\nearrow(y);\;
r_\nearrow(y)r_\vee(x)=r_\searrow(x\nearrow y);\eqno (4.1.8)$$
$$l_\searrow(x* y)=l_\searrow(x)l_\searrow(y);\;\;
r_\searrow(y)l_*(x)=l_\searrow(x)r_\searrow(y);\;\; r_\searrow(y)r_*
(x)=r_\searrow(x\searrow y),\eqno (4.1.9)$$ where
$$x\succ y=x\nearrow y+x\searrow y,
x\prec y=x\nwarrow y+x\swarrow y,\eqno (4.1.10)$$
$$l_\succ=l_\nearrow+l_\searrow,\;
r_\succ=r_\nearrow+r_\searrow,\; l_\prec=l_\nwarrow+l_\swarrow,\;
r_\prec=r_\nwarrow +r_\swarrow,\;\eqno (4.1.11)$$
$$x\vee y=x\searrow
y+x\swarrow y,x\wedge y=x\nearrow y+x\nwarrow y,\eqno (4.1.12)$$
$$l_\vee=l_\searrow+l_\swarrow,\;
r_\vee=r_\searrow+r_\swarrow,\; l_\wedge=l_\nearrow+l_\nwarrow,\;
r_\wedge=r_\nearrow+r_\nwarrow,\;\eqno (4.1.13)$$
$$x*y=x\searrow y+x\nearrow
y+x\nwarrow y+x\swarrow y,\eqno (4.1.14)$$
$$l_*=l_\searrow+l_\nearrow+l_\nwarrow+l_\swarrow,\;
r_*=r_\searrow+r_\nearrow+r_\nwarrow+r_\swarrow.\eqno
(4.1.15)$$}\end{defn}

According to \cite{Sc}, $(l_\searrow, r_\searrow, l_\nearrow,
r_\nearrow, l_\nwarrow,r_\nwarrow, l_\swarrow, r_\swarrow,V)$ is a
bimodule of a quadri-algebra $(A$, $\searrow, \nearrow, \nwarrow,
\swarrow)$ if and only if there exists a quadri-algebra structure
on the direct sum $A\oplus V$ of the underlying vector spaces of
$A$ and $V$ given by ($\forall x,y\in A, u,v\in V$)
$$(x+u)\searrow (y+v)=x\searrow y+l_\searrow(x)v+r_\searrow (y)u,\;\;
(x+u)\nearrow (y+v)=x\nearrow y+l_\nearrow(x)v+r_\nearrow (y)u,$$
$$(x+u)\nwarrow (y+v)=x\nwarrow y+l_\nwarrow(x)v+r_\nwarrow (y)u,\;\;
(x+u)\swarrow (y+v)=x\swarrow y+l_\swarrow(x)v+r_\swarrow (y)u.\eqno
(4.1.16)$$ We denote it by $A\ltimes_{l_\searrow, r_\searrow,
l_\nearrow, r_\nearrow, l_\nwarrow,r_\nwarrow, l_\swarrow,
r_\swarrow}V$.

\begin{prop} Let $(l_\searrow, r_\searrow, l_\nearrow, r_\nearrow,
l_\nwarrow,r_\nwarrow, l_\swarrow, r_\swarrow,V)$ be a bimodule of
a quadri-algebra $(A, \searrow, \nearrow, \nwarrow, \swarrow)$.
Let $(A,\succ,\prec)$ be the associated horizontal dendriform
algebra.

(1) $(l_\searrow, r_\nearrow, l_\swarrow, r_\nwarrow, V)$ is a
bimodule of $(A,\succ$, $\prec)$.

(2)
$(l_\nearrow+l_\searrow,r_\nearrow+r_\searrow,l_\nwarrow+l_\swarrow,r_\nwarrow
+r_\swarrow, V)$ is a bimodule of $(A,\succ,\prec)$.

(3) For any bimodule $(l_\succ,r_\succ,l_\prec,r_\prec, V)$ of
$(A,\succ,\prec)$, $(l_\succ,0,0,r_\succ,0,r_\prec,l_\prec,0, V)$ is
a bimodule of $(A, \searrow, \nearrow, \nwarrow, \swarrow)$.

(4) Both $(l_\searrow, 0,0, r_\nearrow,0,  r_\nwarrow,
l_\swarrow,0,V)$ and $(l_\nearrow+l_\searrow,0,0,
r_\nearrow+r_\searrow,0,r_\nwarrow +r_\swarrow,
l_\nwarrow+l_\swarrow,0, V)$ are bimodules of $(A, \searrow,
\nearrow, \nwarrow, \swarrow)$.

(5) The quadri-algebras $A\ltimes_{l_\searrow, r_\searrow,
l_\nearrow, r_\nearrow, l_\nwarrow,r_\nwarrow, l_\swarrow,
r_\swarrow} V$ and $A\ltimes_{l_\nearrow+l_\searrow,0,0,
r_\nearrow+r_\searrow,0,r_\nwarrow
+r_\swarrow,l_\nwarrow+l_\swarrow, 0}V$  have the same associated
horizontal dendriform algebra
$A\ltimes_{l_\nearrow+l_\searrow,r_\nearrow+r_\searrow,l_\nwarrow+l_\swarrow,r_\nwarrow
+r_\swarrow} V$.\end{prop}

\begin{proof}
(1) follows from the following correspondences of equations:

(3.1.1) $\Longleftrightarrow$ (4.1.3-1)\quad\quad  (3.1.2)
$\Longleftrightarrow$ (4.1.2-2);\quad\quad (3.1.3)
$\Longleftrightarrow$ (4.1.1-3);

(3.1.4) $\Longleftrightarrow$ (4.1.6-1);\quad\quad  (3.1.5)
$\Longleftrightarrow$ (4.1.5-2);\quad\quad (3.1.6)
$\Longleftrightarrow$ (4.1.4-3);

(3.1.7) $\Longleftrightarrow$ (4.1.9-1);\quad\quad  (3.1.8)
$\Longleftrightarrow$ (4.1.8-2);\quad\quad (3.1.9)
$\Longleftrightarrow$ (4.1.7-3).

\noindent (2) follows from the following correspondences of
equations:

(3.1.1) $\Longleftrightarrow$ (4.1.1-1)+(4.1.2-1)+(4.1.3-1);\quad
(3.1.2) $\Longleftrightarrow$ (4.1.1-2)+(4.1.2-2)+(4.1.3-2);

(3.1.3) $\Longleftrightarrow$ (4.1.1-3)+(4.1.2-3)+(4.1.3-3);\quad
(3.1.4) $\Longleftrightarrow$ (4.1.4-1)+(4.1.5-1)+(4.1.6-1);

(3.1.5) $\Longleftrightarrow$ (4.1.4-2)+(4.1.5-2)+(4.1.6-2); \quad
(3.1.6) $\Longleftrightarrow$ (4.1.4-3)+(4.1.5-3)+(4.1.6-3);

(3.1.7) $\Longleftrightarrow$ (4.1.7-1)+(4.1.8-1)+(4.1.9-1);\quad
(3.1.8) $\Longleftrightarrow$ (4.1.7-2)+(4.1.8-2)+(4.1.9-2);

(3.1.9) $\Longleftrightarrow$ (4.1.7-3)+(4.1.8-3)+(4.1.9-3).

\noindent (3) In this case
$r_\searrow=l_\nearrow=l_\nwarrow=r_\swarrow=0, l_\searrow=l_\succ,
r_\nearrow=r_\succ, r_\nwarrow=r_\prec,l_\swarrow=l_\prec$, it is
obvious that the fact that $(l_\succ,r_\succ, l_\prec, r_\prec)$ is
a bimodule of $(A,\succ,\prec)$ corresponds to the equations
appearing in (1) and the other equations hold (both sides are zero).
So (3) holds.

\noindent (4) follows immediately from (1), (2) and (3).

\noindent (5) follows immediately from (4).
\end{proof}

By the above conclusion and Proposition 3.1.2, we have the following
results.

\begin{coro} Let $(l_\searrow, r_\searrow, l_\nearrow, r_\nearrow,
l_\nwarrow,r_\nwarrow, l_\swarrow, r_\swarrow,V)$ be a bimodule of a
quadri-algebra $(A$, $\searrow$, $\nearrow, \nwarrow, \swarrow)$.
Let $(A,*)$ be the associated associative algebra.

(1) $(l_\searrow, r_\nwarrow, V)$,$(l_\nearrow+l_\searrow,r_\nwarrow
+r_\swarrow, V)$,$(l_\searrow+l_\swarrow, r_\nearrow+r_\nwarrow,
V)$,$(l_\nearrow+l_\searrow+l_\nwarrow+l_\swarrow,r_\nearrow+r_\searrow+r_\nwarrow
+r_\swarrow, V)$ are bimodules of $(A,*)$.

(2) For any bimodule $(l,r,V)$ of $(A,*)$, $(l,0,0,0,0,r,0,0)$ is a
bimodule of $(A, \searrow, \nearrow, \nwarrow, \swarrow)$.\end{coro}

\begin{prop} Let $(l_\searrow, r_\searrow, l_\nearrow, r_\nearrow,
l_\nwarrow,r_\nwarrow, l_\swarrow, r_\swarrow,V)$ be a bimodule of
a quadri-algebra $(A, \searrow, \nearrow, \nwarrow, \swarrow)$.
Then $(r_*^*, l_\nwarrow^*,-r_\vee^*,
-l_\prec^*,r_\searrow^*,l_*^*,-r_\succ^*,-l_\wedge^*,A^*)$ is a
bimodule of $(A, \searrow, \nearrow, \nwarrow$, $\swarrow)$. We
call it the dual bimodule of $(l_\searrow, r_\searrow, l_\nearrow,
r_\nearrow, l_\nwarrow,r_\nwarrow, l_\swarrow,
r_\swarrow,V)$.\end{prop}

\begin{proof}
This conclusion can be obtained by a direct checking on equations
(4.1.1)-(4.1.9). We give another approach by using the relations
between bimodules of a quadri-algebra $(A, \searrow$, $ \nearrow,
\nwarrow$, $\swarrow)$ and the associated horizontal dendriform
algebra $(A,\succ,\prec)$ with the (known) dual bimodules of
$(A,\succ,\prec)$. Let $(\bar l_\searrow, \bar r_\searrow, \bar
l_\nearrow, \bar r_\nearrow, \bar l_\nwarrow, \bar r_\nwarrow,
\bar l_\swarrow, \bar r_\swarrow,V^*)$ be the dual bimodule of
$(l_\searrow, r_\searrow, l_\nearrow, r_\nearrow$,
$l_\nwarrow,r_\nwarrow, l_\swarrow, r_\swarrow,V)$. Then by
Proposition 4.1.2,  both
$$(\bar l_\searrow, \bar r_\nearrow, \bar l_\swarrow, \bar r_\nwarrow,
V^*)\;\;{\rm and}\;\;(\bar l_\nearrow+\bar l_\searrow,\bar
r_\nearrow+\bar r_\searrow,\bar l_\nwarrow+\bar l_\swarrow,\bar
r_\nwarrow +\bar r_\swarrow, V^*)$$ are bimodules of
$(A,\succ,\prec)$. On the other hand, since both $(l_\searrow,
r_\nearrow, l_\swarrow, r_\nwarrow, V)$ and
$(l_\nearrow+l_\searrow,r_\nearrow+r_\searrow,l_\nwarrow+l_\swarrow,r_\nwarrow
+r_\swarrow, V)$ are bimodules of $(A,\succ,\prec)$, their dual
bimodules are
$$ (r_\wedge^*,
-l_\swarrow^*, -r_\nearrow^*,l_\vee^*,V^*),\;\;{\rm and}\;\; (r_*^*,
-l_\prec^*,-r_\succ^*,l_*^*,V^*)$$ respectively. Hence we have the
following equations
$$\bar l_\searrow=r_*^*, \bar r_\nearrow=-l_\prec^*, \bar
l_\swarrow=-r_\succ^*, \bar r_\nwarrow^*=l_*^*;$$
$$\bar l_\nearrow+\bar l_\searrow=r_\wedge^*
,\bar r_\nearrow+\bar r_\searrow=-l_\swarrow^*,\bar l_\nwarrow+\bar
l_\swarrow=-r_\nearrow^*,\bar r_\nwarrow +\bar
r_\swarrow=l_\vee^*.$$ Therefore $(r_*^*, l_\nwarrow^*,-r_\vee^*,
-l_\prec^*,r_\searrow^*,l_*^*,-r_\succ^*,-l_\wedge^*,A^*)$ is a
bimodule of $(A, \searrow, \nearrow, \nwarrow, \swarrow)$.
\end{proof}

By Propositions 4.1.2 and 4.1.4, the following conclusion is
obvious.

\begin{coro}   Let $(l_\searrow, r_\searrow, l_\nearrow, r_\nearrow,
l_\nwarrow,r_\nwarrow, l_\swarrow, r_\swarrow,V)$ be a bimodule of a
quadri-algebra $(A$, $\searrow$, $\nearrow, \nwarrow, \swarrow)$.

(1) Both $(r_*^*, 0,0 -l_\prec^*,0,l_*^*,-r_\succ^*,0,V^*)$ and
$(r_\wedge^*, 0,0, -l_\swarrow^*,0, l_\vee^*,-r_\nearrow^*,0,V^*)$
are bimodules of $(A, \searrow, \nearrow, \nwarrow, \swarrow)$.

(2) Both $(r_*^*, -l_\prec^*,-r_\succ^*,l_*^*,V^*)$ and
$(r_\wedge^*, -l_\swarrow^*, -r_\nearrow^*,l_\vee^*,V^*)$ are
bimodules of the associated horizontal dendriform algebra
$(A,\succ,\prec)$.\end{coro}

\begin{exam}{\rm Let $(A, \searrow, \nearrow, \nwarrow, \swarrow)$ be a
quadri-algebra. Then $$(L_\searrow, R_\searrow, L_\nearrow,
R_\nearrow, L_\nwarrow,R_\nwarrow, L_\swarrow,
R_\swarrow,A),\;(L_\searrow, 0,0, R_\nearrow,0,
R_\nwarrow,L_\swarrow, 0,A),\;$$$$(L_\nearrow+L_\searrow,0,0,
R_\nearrow+R_\searrow,0,R_\nwarrow
+R_\swarrow,L_\nwarrow+L_\swarrow, 0, A)$$ are bimodules of $(A,
\searrow, \nearrow$, $\nwarrow, \swarrow)$ and the first one is
called the regular bimodule of $(A, \searrow, \nearrow$, $\nwarrow,
\swarrow)$. On the other hand,
$$(R_*^*, L_\nwarrow^*,-R_\vee^*,
-L_\prec^*,R_\searrow^*,L_*^*,-R_\succ^*,-L_\wedge^*,A^*),\; (R_*^*,
0,0 -L_\prec^*,0,L_*^*,-R_\succ^*,0,A^*),$$
$$(R_\wedge^*, 0,0, -L_\swarrow^*,0, L_\vee^*,-R_\nearrow^*,0,A^*)$$
are bimodules of $(A, \searrow, \nearrow, \nwarrow, \swarrow)$,
too.}\end{exam}

\subsection{ ${\mathcal O}$-operators of quadri-algebras and
$Q$-equation}

\begin{defn}{\rm Let $(l_\searrow, r_\searrow, l_\nearrow, r_\nearrow,
l_\nwarrow,r_\nwarrow, l_\swarrow, r_\swarrow,V)$ be a bimodule of a
quadri-algebra $(A, \searrow, \nearrow, \nwarrow, \swarrow)$. A
linear map $T:V\rightarrow A$ is called an {\it ${\mathcal
O}$-operator of $(A, \searrow, \nearrow, \nwarrow, \swarrow)$
associated to $(l_\searrow, r_\searrow, l_\nearrow, r_\nearrow,
l_\nwarrow,r_\nwarrow, l_\swarrow, r_\swarrow,V)$} if $T$ satisfies
$$T(u)\searrow T(v)=T(l_\searrow(T(u))v+r_\searrow(T(v)u)),\eqno (4.2.1)$$
$$T(u)\nearrow T(v)=T(l_\nearrow(T(u))v+r_\nearrow(T(v)u)),\eqno (4.2.2)$$
$$T(u)\nwarrow T(v)=T(l_\nwarrow(T(u))v+r_\nwarrow(T(v)u)),\eqno (4.2.3)$$
$$T(u)\swarrow T(v)=T(l_\swarrow(T(u))v+r_\swarrow(T(v)u)),\eqno (4.2.4)$$
for any $u,v\in V$.}\end{defn}

The following result is obvious.

\begin{coro}  Let
$T$ be an ${\mathcal O}$-operator of a quadri-algebra $(A, \searrow,
\nearrow, \nwarrow, \swarrow)$ associated to a bimodule
$(l_\searrow, r_\searrow, l_\nearrow, r_\nearrow$,
$l_\nwarrow,r_\nwarrow, l_\swarrow, r_\swarrow,V)$.

(1) $T$ is an $\mathcal O$-operator of the associated horizontal
dendriform algebra $(A,\succ,\prec)$ associated to
$(l_\nearrow+l_\searrow,r_\nearrow+r_\searrow,l_\nwarrow+l_\swarrow,r_\nwarrow
+r_\swarrow, V)$.

(2) $T$ is an $\mathcal O$-operator of the associated associative
algebra $(A,*)$ associated to
$(l_\nearrow+l_\searrow+l_\nwarrow+l_\swarrow,r_\nearrow+r_\searrow+r_\nwarrow
+r_\swarrow, V)$.\end{coro}

\begin{prop} Let $(A, \searrow, \nearrow, \nwarrow, \swarrow)$ be a
quadri-algebra. Let $r\in A\otimes A$  be skew-symmetric. Then $r$
is an ${\mathcal O}$-operator of the quadri-algebra $(A, \searrow,
\nearrow, \nwarrow, \swarrow)$ associated to $(R_*^*,
L_\nwarrow^*,-R_\vee^*$, $-L_\prec^*$,
$R_\searrow^*,L_*^*,-R_\succ^*,-L_\wedge^*,A^*)$ if and only if $r$
satisfies
$$r_{13}\searrow r_{23}=r_{23}*r_{12}-r_{12}\nwarrow r_{13},\eqno
(4.2.5)$$
$$r_{13}\nearrow r_{23}=-r_{23}\vee r_{12}+r_{12}\prec
r_{13},\eqno (4.2.6)$$
$$r_{13}\nwarrow r_{23}=r_{23}\searrow r_{12}-r_{12}*r_{13},\eqno
(4.2.7)$$
$$r_{13}\swarrow r_{23}=-r_{23}\succ r_{12}+r_{12}\wedge r_{13}.\eqno
(4.2.8)$$\end{prop}

\begin{proof}
The conclusion follows from a similar proof as of Proposition
3.4.13.
\end{proof}

\begin{lemma} Let $(A, \searrow, \nearrow, \nwarrow, \swarrow)$ be a
quadri-algebra and $r\in A\otimes A$. Let $r$ be skew-symmetric.

(1) Equation (3.4.17) holds if and only if equation (4.2.8) holds.

(2) Equation (3.4.18) holds if and only if equation (4.2.6) holds.

(3) Equation (3.4.19) holds if and only if equation (4.2.5) holds.

(4) equation (3.4.19) holds if and only if equation (4.2.7)
holds.\end{lemma}

\begin{proof}
Let $\sigma_{123},\sigma_{132}: A\otimes A\otimes A\rightarrow
A\otimes A\otimes A$ be two linear maps given by
$$\sigma_{123}(x\otimes y\otimes z)=z\otimes x\otimes y,\;\;
\sigma_{132}(x\otimes y\otimes z)=y\otimes z\otimes x,\;\;\forall\;
x,y,z\in A,$$ respectively. Then we have the following equations.
\begin{eqnarray*}
\sigma_{132}(r_{13}\succ r_{23}-r_{23}\wedge r_{12}-r_{12}\swarrow
r_{13})&=& r_{32}\succ r_{12}-r_{12}\wedge r_{31}-r_{31}\swarrow
r_{32}\\
&=&-r_{23}\succ r_{12}+r_{12}\wedge r_{13}-r_{13}\swarrow r_{23};\\
\sigma_{123}(r_{13}\prec r_{23}+r_{23}\nearrow r_{12}+r_{12}\vee
r_{23}) &=& r_{21}\prec r_{31}+r_{31}\nearrow r_{23}+r_{23}\vee
r_{21}\\
&=& r_{12}\prec r_{13}-r_{13}\nearrow r_{23}-r_{23}\vee r_{12};\\
\sigma_{132}(r_{13}*r_{23}-r_{23}\nwarrow r_{12}+r_{12}\searrow
r_{13})&=& r_{32}*r_{12}-r_{12}\nwarrow r_{31}+r_{31}\searrow
r_{32}\\
&=&-r_{23}* r_{12}+r_{12}\nwarrow r_{13}+r_{13}\searrow r_{23};\\
\end{eqnarray*}
\begin{eqnarray*}
\sigma_{123}(r_{13}*r_{23}-r_{23}\nwarrow r_{12}+r_{12}\searrow
r_{13}) &=& r_{21}* r_{31}-r_{31}\nwarrow r_{23}+r_{23}\searrow
r_{21}\\
&=& r_{12}*r_{13}+r_{13}\nwarrow r_{23}-r_{23}\searrow r_{12}.
\end{eqnarray*}
Therefore the conclusion follows immediately.
\end{proof}

By Propositions 3.4.13 and 4.2.3 and Lemma 4.2.4, we have the
following conclusion.

\begin{coro} Let $(A, \searrow, \nearrow, \nwarrow, \swarrow)$ be a
quadri-algebra and $r\in A\otimes A$. Let $r$ be skew-symmetric.
Then the following conditions are equivalent.

(1) $r$ is an ${\mathcal O}$-operator of the associated horizontal
dendriform algebra $(A,\succ,\prec)$ associated to the bimodule
$(R_\nearrow^*+R_\nwarrow^*,-L_\swarrow^*,-R_\nearrow^*,
L_\searrow^*+L_\swarrow^*,A^*)$.

(2) $r$ is an ${\mathcal O}$-operator of the quadri-algebra $(A,
\searrow, \nearrow, \nwarrow, \swarrow)$ associated to the bimodule
$(R_*^*, L_\nwarrow^*,-R_\vee^*,
-L_\prec^*,R_\searrow^*,L_*^*,-R_\succ^*,-L_\wedge^*,A^*).$

(3) $r$ satisfies equations (3.4.17) and (3.4.18) in $(A, \searrow,
\nearrow, \nwarrow, \swarrow)$.

(4) $r$ satisfies equations (4.2.6) and (4.2.8) in $(A, \searrow,
\nearrow, \nwarrow, \swarrow)$.\end{coro}

\begin{defn}{\rm Let $(A, \searrow, \nearrow, \nwarrow, \swarrow)$ be a
quadri-algebra and $r\in A\otimes A$. A set of equations (3.4.17)
and (3.4.18) is called {\it $Q$-equation in $(A, \searrow, \nearrow,
\nwarrow, \swarrow)$}.}\end{defn}

\begin{remark}{\rm Due to Corollary 4.2.5, it is reasonable to regard the
$Q$-equation (a set of equations) in a quadri-algebra as an analogue
of the classical Yang-Baxter equation in a Lie algebra.}\end{remark}

With a similar discussion as in subsection 3.3, we have the
following results (see Proposition 3.3.5, Corollaries 3.3.6 and
3.3.8).

\begin{theorem} Let $(l_\searrow, r_\searrow, l_\nearrow, r_\nearrow,
l_\nwarrow,r_\nwarrow, l_\swarrow, r_\swarrow,V)$ be a bimodule of a
quadri-algebra $(A$, $\searrow$, $\nearrow, \nwarrow, \swarrow)$.
Let $(r_*^*, l_\nwarrow^*,-r_\vee^*,
-l_\prec^*,r_\searrow^*,l_*^*,-r_\succ^*,-l_\wedge^*,A^*)$ be the
dual bimodule given by Proposition 4.1.4. Let $T:V\rightarrow A$ be
a linear map identified as an element in $A\otimes V^*$ which is in
the underlying vector space of
$$(A\ltimes_{r_*^*,
l_\nwarrow^*,-r_\vee^*,
-l_\prec^*,r_\searrow^*,l_*^*,-r_\succ^*,-l_\wedge^*}V^*)\otimes
(A\ltimes_{r_*^*, l_\nwarrow^*,-r_\vee^*,
-l_\prec^*,r_\searrow^*,l_*^*,-r_\succ^*,-l_\wedge^*}V^*).$$ Then
$r=T-\sigma(T)$ is a skew-symmetric solution of $Q$-equation in the
quadri-algebra\\ $A\ltimes_{r_*^*, l_\nwarrow^*,-r_\vee^*,
-l_\prec^*,r_\searrow^*,l_*^*,-r_\succ^*,-l_\wedge^*}V^*$ if and
only if $T$ is an ${\mathcal O}$-operator  of the quadri-algebra
$(A$, $\searrow, \nearrow, \nwarrow, \swarrow)$ associated to
$(l_\searrow, r_\searrow, l_\nearrow, r_\nearrow,
l_\nwarrow,r_\nwarrow, l_\swarrow, r_\swarrow,V)$.\end{theorem}

\begin{coro} Let $(A,\succ, \prec)$ be a dendriform algebra. Let
$(l_\succ,r_\succ,l_\prec,r_\prec, V)$ be a bimodule of $A$ and
$(r_\succ^*+r_\prec^*,-l_\prec^*,-r_\succ^*,l_\succ^*+l_\prec^*,V^*)$
be the dual bimodule given in Proposition 3.1.3. Suppose that
$T:V\rightarrow A$ is an ${\mathcal O}$-operator of
$(A,\succ,\prec)$ associated to $(l_\succ,r_\succ,l_\prec,r_\prec,
V)$. Then $r=T-\sigma(T)$ is a skew-symmetric solution of
$Q$-equation in the quadri-algebra
$$T(V)\ltimes_{r_\succ^*+r_\prec^*,0,0,-l_\prec^*,0,l_\succ^*+l_\prec^*,-r_\succ^*,0,
}V^*,$$ where $T(V)\subset A$ is a quadri-algebra given by equation
(3.4.9) and
$(r_\succ^*+r_\prec^*,0,0,-l_\prec^*,0,l_\succ^*+l_\prec^*,-r_\succ^*,0,
V^* )$ is its bimodule since its associated horizontal dendriform
algebra $T(V)$ is a dendriform subalgebra of $(A,\succ,\prec)$, and
$T$ can be identified as an element in $T(V)\otimes V^*$ which is in
the underlying vector space of $$
(T(V)\ltimes_{r_\succ^*+r_\prec^*,0,0,-l_\prec^*,0,l_\succ^*+l_\prec^*,-r_\succ^*,0,
}V^*)\otimes
(T(V)\ltimes_{r_\succ^*+r_\prec^*,0,0,-l_\prec^*,0,l_\succ^*+l_\prec^*,-r_\succ^*,0,
}V^*). $$\end{coro}

\begin{coro}{\rm (cf. Corollary 4.4.13)}\quad Let $(A, \searrow, \nearrow,
\nwarrow, \swarrow)$ be a quadri-algebra. Then $r$ given by equation
(2.2.9) is a skew-symmetric solution of $Q$-equation in the
quadri-algebra $A\ltimes_{R_\wedge^*, 0,0, -L_\swarrow^*,0,
L_\vee^*,-R_\nearrow^*,0}A^*$, where $\{e_1,\cdots, e_n\}$ is a
basis of $A$ and $\{e_1^*,\cdots, e_n^*\}$ is its dual
basis.\end{coro}

\subsection{Bilinear forms on quadri-algebras and $Q$-equation}

In this subsection, we consider the (symmetric and skew-symmetric)
bilinear forms on quadri-algebras as we have done in subsections 2.2
and 3.2.

\begin{lemma} Let $(A, \searrow, \nearrow, \nwarrow, \swarrow)$ be a
quadri-algebra. Let ${\mathcal B}$ be a symmetric bilinear form on
$A$. Set
$$\mathcal B(x,y\succ z)=\mathcal B(x\wedge y, z),\eqno (4.3.1)$$
$$\mathcal B(y,z\prec x)=\mathcal B(x\vee y, z),\eqno (4.3.2)$$
$$\mathcal B(x\searrow y,z)=\mathcal B(x,y\nwarrow z),\eqno
(4.3.3)$$
$$\mathcal B(x\nearrow y,z)+\mathcal B(z\swarrow x,y)+\mathcal
B(y*z,x)=0,\eqno (4.3.4)$$
$$\mathcal B(x\searrow y,z)+\mathcal B(y\nearrow z, x)-\mathcal B
(z\succ x, y)=0,\eqno (4.3.5)$$
$$\mathcal B(x\nwarrow y,z)+\mathcal B(z\swarrow x, y)-\mathcal
B(y\prec z, x)=0,\eqno (4.3.6)$$ for any $x,y,z\in A$. Then any two
equations of the three equations in the following sets can imply the
third equation.

(1) Equations (3.4.14), (3.4.15) and (4.3.1);

(2) Equations (3.4.13), (3.4.16) and (4.3.2);

(3) Equations (3.4.13), (3.4.15) and (4.3.3);

(4) Equations (3.4.14), (3.4.16) and (4.3.4);

(5) Equations (3.4.13), (3.4.14) and (4.3.5);

(6) Equations (3.4.15), (3.4.16) and (4.3.6).\end{lemma}

\begin{proof}
It is straightforward.
\end{proof}

\begin{defn} {\rm Let $(A, \searrow, \nearrow, \nwarrow, \swarrow)$ be a
quadri-algebra. A symmetric bilinear form ${\mathcal B}$ on $A$ is
called {\it invariant} if $\mathcal B$ satisfies equations
(3.4.13)-(3.4.16).}\end{defn}

\begin{coro} Any symmetric invariant bilinear form on a
quadri-algebra is a 2-cocycle of the associated  horizontal
dendriform algebra.\end{coro}

\begin{proof} It follows immediately from
the sum of equations (3.4.13)-(3.4.16).
\end{proof}

By Proposition 3.4.11 and  Corollary 4.3.3, we have the following
result.

\begin{coro} Let $(A,\succ,\prec)$ be a dendriform algebra and $\mathcal B$
be a nondegenerate symmetric bilinear form. Then $\mathcal B$ is a
2-cocycle of $(A,\succ,\prec)$ if and only if $\mathcal B$ is
invariant on the compatible quadri-algebra structure given by
equations (3.4.13)-(3.4.16) which $(A,\succ,\prec)$ is the
associated horizontal dendriform algebra.\end{coro}

\begin{remark}{\rm In Remark 4.2.7, we have seen that the $Q$-equation in a
quadri-algebra is an analogue of the classical Yang-Baxter equation
as an $\mathcal O$-operator of the quadri-algebra (or equivalently,
its associated horizontal dendriform algebra) associated to certain
dual bimodule. On the other hand, similar to the associative
Yang-Baxter equation in an associative algebra coming from the
double construction of the nondegenerate symmetric invariant
bilinear forms on associative algebras and $D$-equation in a
dendriform algebra coming from the double constructions of the
nondegenerate (skew-symmetric) invariant bilinear forms on
dendriform algebras (or equivalently, the double construction of the
nondegenerate Connes cocycles on associative algebras), it is quite
reasonable to believe that the $Q$-equation in a quadri-algebra
should be related to certain ``double construction" of the
nondegenerate (symmetric) invariant bilinear forms on
quadri-algebras (or equivalently, the double construction of the
nondegenerate 2-cocycles on dendriform algebras). In fact, such a
conclusion has been proved in \cite{Ni}.}\end{remark}

Next we turn to the study of skew-symmetric bilinear forms on
quadri-algebras.

\begin{theorem}  Let $(A, \searrow, \nearrow, \nwarrow, \swarrow)$ be a
quadri-algebra and $r\in A\otimes A$. Suppose that $r$ is
skew-symmetric and nondegenerate. Then $r$ is a solution of
$Q$-equation in $A$ if and only if the inverse of the isomorphism
$A^*\rightarrow A$ induced by $r$, regarded as a bilinear form
$\omega$ on $A$ (that is, $\omega(x,y)=\langle r^{-1}x,y\rangle $
for any $x,y\in A$) satisfies (for any $x,y,z\in A$)
$$\omega(z,x\succ y)-\omega(x,y\swarrow z)+\omega(y,z\wedge
x)=0,\; \omega(z,x\prec y)+\omega(x,y\vee z)-\omega(y,z\nearrow
x)=0. \eqno (4.3.7)$$\end{theorem}

\begin{proof}
Let $r=\sum\limits_i a_i\otimes b_i$. Since $r$ is skew-symmetric,
we have $\sum\limits_i a_i\otimes b_i=-\sum\limits_ib_i\otimes a_i$.
Therefore $r(v^*)=-\sum\limits_i\langle v^*,a_i\rangle
b_i=\sum\limits_i\langle v^*,b_i\rangle  a_i$ for any $v^*\in A^*$.
Since $r$ is nondegenerate, for any $x,y,z\in A$, there exist
$u^*,v^*,w^*\in A^*$ such that $x=r(u^*),y=r(v^*),z=r(w^*)$.
Therefore
\begin{eqnarray*}
\langle u^*\otimes v^*\otimes w^*, r_{13}\succ r_{23}\rangle
&=&\sum_{i,j}\langle u^*,a_i\rangle\langle v^*,a_j\rangle\langle
w^*,b_i\succ b_j\rangle\\
&=&\langle r(u^*)\succ r(v^*), w^*\rangle=\omega(z,x\succ y);\\
\langle u^*\otimes v^*\otimes w^*, r_{23}\wedge r_{12}\rangle
&=&\sum_{i,j}\langle u^*,a_j\rangle\langle v^*,a_i\wedge
b_j\rangle\langle
w^*,b_i\succ b_j\rangle\\
&=&-\langle r(w^*)\wedge r(u^*), v^*\rangle=-\omega(y,z\wedge x);\\
\langle u^*\otimes v^*\otimes w^*, r_{12}\swarrow r_{23}\rangle
&=&\sum_{i,j}\langle u^*,a_i\swarrow a_j\rangle\langle
v^*,b_i\rangle\langle
w^*,b_j\succ b_j\rangle\\
&=&\langle r(v^*)\swarrow r(w^*), u^*\rangle=\omega(x,y\swarrow z).
\end{eqnarray*}
So $r$ satisfies equation (3.4.17) if and only if $\omega$ satisfies
$$\omega(z,x\succ y)-\omega(x,y\swarrow z)+\omega(y,z\wedge
x)=0,\;\;\forall x,y,z\in A.$$ Similarly, $r$ satisfies equation
(3.4.18) if and only if $\omega$ satisfies $$ \omega(z,x\prec
y)+\omega(x,y\vee z)-\omega(y,z\nearrow x)=0,\;\forall x,y,z\in A.$$
\end{proof}

\begin{defn}{\rm  Let $(A, \searrow, \nearrow, \nwarrow, \swarrow)$ be a
quadri-algebra. A skew-symmetric bilinear form $\omega$ on A is
called a {\it 2-cocycle}   if $\omega$ satisfies equation
(4.3.7).}\end{defn}

By Corollary 4.2.10 and Theorem 4.3.6, we have the following
conclusion.

\begin{coro} Let $(A, \searrow, \nearrow, \nwarrow, \swarrow)$ be a
quadri-algebra. Then the skew-symmetric solution of $Q$-equation in
the quadri-algebra $A\ltimes_{R_\wedge^*, 0,0, -L_\swarrow^*,0,
L_\vee^*,-R_\nearrow^*,0}A^*$ given by equation (2.2.9) induces a
natural 2-cocycle $\omega$  on $A\ltimes_{R_\wedge^*, 0,0,
-L_\swarrow^*,0, L_\vee^*,-R_\nearrow^*,0}A^*$ by $r^{-1}: A\oplus
A^*\rightarrow (A\oplus A^*)^*$, which is given by equation
(2.2.10).\end{coro}

\subsection{${\mathcal O}$-operators of quadri-algebras and
octo-algebras}

\begin{defn}{\rm (\cite{Le3})\quad Let $A$ be a vector space
with eight bilinear products denoted by
$$\searrow_1, \searrow_2, \nearrow_1, \nearrow_2, \nwarrow_1, \nwarrow_2,
\swarrow_1, \swarrow_2: A\otimes A\rightarrow A.$$ $(A, \searrow_1,
\searrow_2, \nearrow_1, \nearrow_2, \nwarrow_1, \nwarrow_2,
\swarrow_1, \swarrow_2)$ is called an {\it octo-algebra} if for any
$x,y,z\in A$,

{\small
$$(x\nwarrow_1 y)\nwarrow_1 z=x\nwarrow_1 (y*
z),(x\nearrow_1 y)\nwarrow_1 z=x\nearrow_1 (y\ll z),(x\wedge_1
y)\nearrow_1 z=x\nearrow_1 (y\gg z),\eqno (4.4.1)$$
$$(x\swarrow_1 y)\nwarrow_1 z=x\swarrow_1 (y\bigwedge z), (x\searrow_1
y)\nwarrow_1 z=x\searrow_1 (y\nwarrow_{12} z),(x\vee_1 y)\nearrow_1
z=x\searrow_1 (y\nearrow_{12} z),\eqno (4.4.2)$$
$$(x\prec_1 y)\swarrow_1 z=x\swarrow_1 (y\bigvee z),(x\succ_1
y)\swarrow_{12} z=x\searrow_1 (y\swarrow_{12} z),(x\Sigma_1
y)\searrow_1 z=x\searrow_1 (y\searrow_{12} z),\eqno(4.4.3)$$
$$(x\nwarrow_2 y)\nwarrow_1 z=x\nwarrow_2 (y\Sigma_1
z),(x\nearrow_2 y)\nwarrow_1 z=x\nearrow_2 (y\prec_1 z),(x\wedge_2
y)\nearrow_1 z=x\nearrow_2 (y\succ_1 z),\eqno (4.4.4)$$
$$(x\swarrow_2 y)\nwarrow_1 z=x\swarrow_2 (y\wedge_1 z),
(x\searrow_2 y)\nwarrow_1 z=x\searrow_2 (y\nwarrow_{1} z),(x\vee_2
y)\nearrow_1 z=x\searrow_2 (y\nearrow_{1} z),\eqno (4.4.5)$$
$$(x\prec_2 y)\swarrow_1 z=x\swarrow_2 (y\vee_1 z),(x\succ_2
y)\swarrow_{1} z=x\searrow_2 (y\swarrow_{1} z),(x\Sigma_2
y)\searrow_1 z=x\searrow_2 (y\searrow_{1} z),\eqno(4.4.6)$$
$$(x\nwarrow_{12} y)\nwarrow_2 z=x\nwarrow_2 (y\Sigma_2
z),(x\nearrow_{12} y)\nwarrow_2 z=x\nearrow_2 (y\prec_2
z),(x\bigwedge y)\nearrow_2 z=x\nearrow_2 (y\succ_2 z),\eqno
(4.4.7)$$
$$(x\swarrow_{12} y)\nwarrow_2 z=x\swarrow_2 (y\wedge_2 z),
(x\searrow_{12} y)\nwarrow_2 z=x\searrow_2 (y\nwarrow_{2}
z),(x\bigvee y)\nearrow_2 z=x\searrow_2 (y\nearrow_{2} z),\eqno
(4.4.8)$$
$$(x\ll y)\swarrow_2 z=x\swarrow_2 (y\vee_2 z),(x\gg
y)\swarrow_{2} z=x\searrow_2 (y\swarrow_{2} z),(x* y)\searrow_2
z=x\searrow_2 (y\searrow_{2} z),\eqno(4.4.9)$$}  where
$$x\succ_i y=x\nearrow_i y+x\searrow_i y,\;
x\prec_i y=x\nwarrow_i y+x\swarrow_i y,\;\;i=1,2;\eqno (4.4.10)$$
$$x\vee_i y=x\searrow_i
y+x\swarrow_i y,\;x\wedge_i y=x\nearrow_i y+x\nwarrow_i
y,\;\;i=1,2;\eqno(4.4.11)$$
$$x\bigvee y= x\vee_1y+x\vee_2y,\; x\bigwedge
y=x\wedge_1y+x\wedge_2y;\eqno (4.4.12)$$
$$x\gg y=x\succ_1y+x\succ_2y,\;x\ll y =x\prec_1y+x\prec_2y;\eqno (4.4.13)$$
$$x\circ_{12}y=x\circ_1 y+x\circ_2y,\;\;{\rm with}\;\;\circ\in
\{\searrow,\nearrow, \nwarrow,\searrow\};\eqno (4.4.14)$$
$$x\Sigma_1 y=x\vee_1 y+x\wedge_1 y=x\succ_1 y+x\prec_1y,\;\;x\Sigma_2 y=x\vee_2 y+x\wedge_2 y=x\succ_2
y+x\prec_2y;\eqno (4.4.15)$$ and
$$x*y=x\Sigma_1y +x\Sigma_2 y=\sum_{i=1}^2 (x\searrow_i
y+x\nearrow_i y+x\nwarrow_i y+x\swarrow_i y).\eqno
(4.4.16)$$}\end{defn}

\begin{prop}{\rm (\cite{Le3})}\quad Let $(A, \searrow_1, \searrow_2, \nearrow_1,
\nearrow_2, \nwarrow_1, \nwarrow_2, \swarrow_1, \swarrow_2)$ be an
octo-algebra.

(1) The product given by equation (4.4.14) defines a
quadri-algebra
$(A,\searrow_{12},\nearrow_{12},\nwarrow_{12},\swarrow_{12}$)
which is called the associated depth quadri-algebra of $(A,
\searrow_1, \searrow_2, \nearrow_1, \nearrow_2, \nwarrow_1,
\nwarrow_2, \swarrow_1, \swarrow_2)$. And $(A, \searrow_1,
\searrow_2, \nearrow_1, \nearrow_2, \nwarrow_1, \nwarrow_2,
\swarrow_1, \swarrow_2)$ is called a compatible octo-algebra
structure on the (depth) quadri-algebra
$(A,\searrow_{12},\nearrow_{12},\nwarrow_{12},\swarrow_{12})$.

(2) The product given by equation (4.4.10) defines a quadri-algebra
$(A,\succ_2,\succ_1,\prec_1,\prec_2$) which is called the associated
vertical quadri-algebra of $(A, \searrow_1, \searrow_2, \nearrow_1,
\nearrow_2, \nwarrow_1, \nwarrow_2, \swarrow_1, \swarrow_2)$. And
$(A, \searrow_1, \searrow_2, \nearrow_1, \nearrow_2, \nwarrow_1,
\nwarrow_2, \swarrow_1, \swarrow_2)$ is called a compatible
octo-algebra structure on the (vertical) quadri-algebra
$(A,\succ_2,\succ_1,\prec_1,\prec_2)$.

(3) The product given by equation (4.4.11) defines a
quadri-algebra $(A,\vee_2,\wedge_2,\wedge_1,\vee_1$) which is
called the associated horizontal quadri-algebra of $(A,
\searrow_1, \searrow_2, \nearrow_1, \nearrow_2, \nwarrow_1,
\nwarrow_2, \swarrow_1, \swarrow_2)$. And $(A, \searrow_1,
\searrow_2, \nearrow_1, \nearrow_2, \nwarrow_1, \nwarrow_2,
\swarrow_1, \swarrow_2)$ is called a compatible octo-algebra
structure on the horizontal quadri-algebra
$(A,\vee_2,\wedge_2,\wedge_1,\vee_1$).

(4) The product given by equation (4.4.12) defines a dendriform
algebra $(A,\bigvee,\bigwedge)$. It is the associated vertical
dendriform algebra of both the quadri-algebras
$(A,\searrow_{12},\nearrow_{12},\nwarrow_{12},\swarrow_{12})$ and
$(A,\vee_2,\wedge_2,\wedge_1,\vee_1$).

(5)  The product given by equation (4.4.13) defines a dendriform
algebra $(A,\gg,\ll)$. It is the associated horizontal dendriform
algebra of both the quadri-algebras
$(A,\searrow_{12},\nearrow_{12},\nwarrow_{12},\swarrow_{12})$ and
$(A,\succ_2,\succ_1,\prec_1,\prec_2)$.

(6)  The product given by equation (4.4.15) defines a dendriform
algebra $(A,\Sigma_2,\Sigma_1)$. It is the associated horizontal
dendriform algebra of the quadri-algebra
$(A,\vee_2,\wedge_2,\wedge_1,\vee_1$) and the associated vertical
dendriform algebra of the quadri-algebra
$(A,\succ_2,\succ_1,\prec_1,\prec_2)$.

(7)  The product given by equation (4.4.16) defines an associative
algebra $(A,*)$ which is called the associated associative algebra
of $(A, \searrow_1, \searrow_2, \nearrow_1, \nearrow_2, \nwarrow_1,
\nwarrow_2, \swarrow_1, \swarrow_2)$. And $(A, \searrow_1,
\searrow_2$, $\nearrow_1, \nearrow_2, \nwarrow_1, \nwarrow_2,
\swarrow_1, \swarrow_2)$ is called a compatible octo-algebra
structure on the associative algebra $(A,*)$.\end{prop}

For brevity, we pay our main attention to the study of the
associated depth quadri-algebras of the octo-algebras. The
corresponding study on the associated vertical and horizontal
quadri-algebras are completely similar.

\begin{prop} Let $A$ be a vector space  with
eight bilinear products denoted by $\searrow_1, \searrow_2$, $
\nearrow_1, \nearrow_2, \nwarrow_1, \nwarrow_2, \swarrow_1,
\swarrow_2: A\otimes A\rightarrow A$. Then $(A, \searrow_1,
\searrow_2, \nearrow_1, \nearrow_2, \nwarrow_1, \nwarrow_2$,
$\swarrow_1, \swarrow_2)$ is an octo-algebra if and only if
$(A,\searrow_{12},\nearrow_{12},\nwarrow_{12},\swarrow_{12})$
defined by equation (4.4.14) is a quadri-algebra and
$(L_{\searrow_2}, R_{\searrow_1}, L_{\nearrow_2}, R_{\nearrow_1},
L_{\nwarrow_2}, R_{\nwarrow_1}, L_{\swarrow_2}, R_{\swarrow_1}, A)$
is a bimodule.\end{prop}

\begin{proof}
The conclusions can be obtained from the following correspondence by
substituting $(L_{\searrow_2}, R_{\searrow_1}, L_{\nearrow_2},
R_{\nearrow_1},L_{\nwarrow_2}, R_{\nwarrow_1}, L_{\swarrow_2},
R_{\swarrow_1},  A)$ into equations (4.1.1)-(4.1.9).

(4.1.1-1) $\Longleftrightarrow$ (4.4.7-1);\quad (4.1.1-2)
$\Longleftrightarrow$ (4.4.4-1);\quad (4.1.1-3)
$\Longleftrightarrow$ (4.4.1-1);

(4.1.2-1) $\Longleftrightarrow$ (4.4.8-1);\quad (4.1.2-2)
$\Longleftrightarrow$ (4.4.5-1);\quad (4.1.2-3)
$\Longleftrightarrow$ (4.4.2-1);

(4.1.3-1) $\Longleftrightarrow$ (4.4.9-1);\quad (4.1.3-2)
$\Longleftrightarrow$ (4.4.6-1);\quad (4.1.3-3)
$\Longleftrightarrow$ (4.4.3-1);

(4.1.4-1) $\Longleftrightarrow$ (4.4.7-2);\quad (4.1.4-2)
$\Longleftrightarrow$ (4.4.4-2);\quad (4.1.4-3)
$\Longleftrightarrow$ (4.4.1-2);

(4.1.5-1) $\Longleftrightarrow$ (4.4.8-2);\quad (4.1.5-2)
$\Longleftrightarrow$ (4.4.5-2);\quad (4.1.5-3)
$\Longleftrightarrow$ (4.4.2-2);

(4.1.6-1) $\Longleftrightarrow$ (4.4.9-2);\quad (4.1.6-2)
$\Longleftrightarrow$ (4.4.6-2);\quad (4.1.6-3)
$\Longleftrightarrow$ (4.4.3-2);

(4.1.7-1) $\Longleftrightarrow$ (4.4.7-3);\quad (4.1.7-2)
$\Longleftrightarrow$ (4.4.4-3);\quad (4.1.1-3)
$\Longleftrightarrow$ (4.4.1-3);

(4.1.8-1) $\Longleftrightarrow$ (4.4.8-3);\quad (4.1.8-2)
$\Longleftrightarrow$ (4.4.5-3);\quad (4.1.8-3)
$\Longleftrightarrow$ (4.4.2-3);

(4.1.9-1) $\Longleftrightarrow$ (4.4.9-3);\quad (4.1.9-2)
$\Longleftrightarrow$ (4.4.6-3);\quad (4.1.9-3)
$\Longleftrightarrow$ (4.4.3-3).
\end{proof}

\begin{coro}  Let $(A, \searrow_1, \searrow_2, \nearrow_1, \nearrow_2,
\nwarrow_1, \nwarrow_2, \swarrow_1, \swarrow_2)$ be an octo-algebra.

(1) $(L_{\searrow_2}, R_{\nearrow_1}, L_{\swarrow_2},
R_{\nwarrow_1}, A)$ is a bimodule of the dendriform algebra
$(A,\gg,\ll)$.

(2) $(L_{\searrow_2}, R_{\nwarrow_1}, A)$ is a bimodule of the
associated associative algebra $(A,*)$.\end{coro}

\begin{proof}
(1) follows from Propositions 4.1.2, 4.4.2 and 4.4.3. (2) follows
from Corollary 4.1.3 and Proposition 4.4.3.
\end{proof}

\begin{coro} Let $(A, \searrow_1, \searrow_2, \nearrow_1, \nearrow_2,
\nwarrow_1, \nwarrow_2, \swarrow_1, \swarrow_2)$ be an octo-algebra
and $(A,\searrow_{12}$, $\nearrow_{12},\nwarrow_{12},\swarrow_{12}$)
be the associated depth quadri-algebra. Then
$$(L_{\searrow_{12}}, R_{\searrow_{12}},
L_{\nearrow_{12}}, R_{\nearrow_{12}},
L_{\nwarrow_{12}},R_{\nwarrow_{12}}, L_{\swarrow_{12}},
R_{\swarrow_{12}},A),\;(L_{\searrow_{12}}, 0,0, R_{\nearrow_{12}},0,
 R_{\nwarrow_{12}},L_{\swarrow_{12}},
0,A),\;$$$$(L_{\nearrow_{12}}+L_{\searrow_{12}},0,0,
R_{\nearrow_{12}}+R_{\searrow_{12}},0,R_{\nwarrow_{12}}
+R_{\swarrow_{12}},L_{\nwarrow_{12}}+L_{\swarrow_{12}}, 0, A),$$
$$(L_{\searrow_2}, R_{\searrow_1}, L_{\nearrow_2}, R_{\nearrow_1},
 L_{\nwarrow_2}, R_{\nwarrow_1},L_{\swarrow_2}, R_{\swarrow_1},
A),\;(L_{\searrow_{2}}, 0,0, R_{\nearrow_{1}},0,
R_{\nwarrow_{1}},L_{\swarrow_{2}}, 0,A),$$
$${\rm and}\;\;(L_{\nearrow_{2}}+L_{\searrow_{2}},0,0,
R_{\nearrow_{1}}+R_{\searrow_{1}},0,R_{\nwarrow_{1}}
+R_{\swarrow_{1}}, L_{\nwarrow_{2}}+L_{\swarrow_{2}},0, A),$$ are
bimodules of
$(A,\searrow_{12},\nearrow_{12},\nwarrow_{12},\swarrow_{12})$. On
the other hand,
$$(R_*^*, L_{\nwarrow_{12}}^*,-R_{\bigvee}^*,
-L_{\ll}^*,R_{\searrow_{12}}^*,L_*^*,-R_{\gg}^*,-L_{\bigwedge}^*,A^*),\;
(R_*^*, 0,0 -L_\ll^*,0,L_*^*,-R_\gg^*,0,A^*),$$
$$(R_{\bigwedge}^*, 0,0, -L_{\swarrow_{12}}^*,0, L_{\bigvee}^*,-R_{\nearrow_{12}}^*,0,A^*),$$
$$(R_{\Sigma_1}^*, L_{\nwarrow_{2}}^*,-R_{\vee_1}^*,
-L_{\prec_2}^*,R_{\searrow_{1}}^*,L_{\Sigma_2}^*,-R_{\succ_1}^*,-L_{\wedge_2}^*,A^*),\;
(R_{\Sigma_1}^*, 0,0
-L_{\prec_2}^*,0,L_{\Sigma_2}^*,-R_{\succ_1}^*,0,A^*),$$
$${\rm and}\;\;
(R_{\wedge_1}^*, 0,0, -L_{\swarrow_{2}}^*,0,
L_{\vee_2}^*,-R_{\nearrow_{1}}^*,0,A^*)$$ are bimodules of
$(A,\searrow_{12},\nearrow_{12},\nwarrow_{12},\swarrow_{12})$,
too.\end{coro}

\begin{proof}
It follows from Propositions 4.1.2 and 4.1.4, Example 4.1.6 and
Proposition 4.4.3.
\end{proof}

\begin{prop} Let
$T:V\rightarrow A$ be an ${\mathcal O}$-operator of a quadri-algebra
$(A, \searrow, \nearrow, \nwarrow, \swarrow)$ associated to a
bimodule $(l_\searrow, r_\searrow, l_\nearrow$, $r_\nearrow,
l_\nwarrow,r_\nwarrow, l_\swarrow, r_\swarrow,V)$. Then there exists
an octo-algebra structure on $V$ given by
$$u\searrow_1 v=r_\searrow(T(v))u,\;u\searrow_2 v=l_\searrow
(T(u))v,\;u\nearrow_1 v=r_\nearrow(T(v))u,\;u\nearrow_2 v=l_\nearrow
(T(u))v,\;$$
$$u\nwarrow_1 v=r_\nwarrow(T(v))u,\;u\nwarrow_2 v=l_\nwarrow
(T(u))v,\;u\swarrow_1 v=r_\swarrow(T(v))u,\;u\swarrow_2 v=l_\swarrow
(T(u))v,\eqno (4.4.17)$$ for any $u,v\in V$. Therefore there exists
a quadri-algebra structure on $V$ given by equation (4.4.14) and $T$
is a homomorphism of quadri-algebras. Furthermore, $T(V)=\{T(v)|v\in
V\}\subset A$ is a quadri-subalgebra of $A$ and there is an induced
octo-algebra structure on $T(V)$ given by
$$T(u)\searrow_1 T(v)=T(u\searrow_1 v),\;T(u)\searrow_2 T(v)=T(u\searrow_2 v),\;
T(u)\nearrow_1 T(v)=T(u\nearrow_1 v),\; $$$$T(u)\nearrow_2
T(v)=T(u\nearrow_2 v),\;T(u)\nwarrow_1 T(v)=T(u\nwarrow_1
v),\;T(u)\nwarrow_2 T(v)=T(u\nwarrow_2 v),\;$$ $$T(u)\swarrow_1
T(v)=T(u\swarrow_1 v),\;T(u)\swarrow_2 T(v)=T(u\swarrow_2
v),\;,\;\forall u,v\in V.\eqno (4.4.18)$$ Moreover, its
corresponding associated depth quadri-algebra structure on $T(V)$
given by equation (4.4.14) is just the quadri-subalgebra structure
of $(A, \searrow, \nearrow, \nwarrow, \swarrow)$ and $T$ is a
homomorphism of quadri-algebras.\end{prop}

\begin{proof}
The proof is similar as of Proposition 3.4.6, where a similar
correspondence is given in the proof of Proposition 4.4.3.
\end{proof}

\begin{defn}{\rm (\cite{Le3})\quad  Let $(A, \searrow, \nearrow, \nwarrow,
\swarrow)$ be a quadri-algebra. An $\mathcal O$-operator $R$ of $(A,
\searrow, \nearrow, \nwarrow, \swarrow)$ associated to the regular
bimodule $(L_\searrow, R_\searrow, L_\nearrow, R_\nearrow,
L_\nwarrow, R_\nwarrow, L_\swarrow, R_\swarrow, A)$ is called a
Rota-Baxter operator on $(A, \searrow, \nearrow, \nwarrow,
\swarrow)$, that is, $R$ satisfies (for any $x,y\in A$)
$$R(x\searrow y)=R(R(x)\searrow y+x\searrow R(y)),\;
R(x\nearrow y)=R(R(x)\nearrow y+x\nearrow R(y)),$$
$$R(x\nwarrow y)=R(R(x)\nwarrow y+x\nwarrow R(y)),\;
R(x\swarrow y)=R(R(x)\swarrow y+x\swarrow R(y)).\eqno
(4.4.19)$$}\end{defn}

By Proposition 4.4.6, the following conclusion follows immediately.

\begin{coro}{\rm (\cite{Le3}, Proposition 5.12)}\quad Let $(A, \searrow,
\nearrow, \nwarrow, \swarrow)$ be a quadri-algebra and $R$ be a
Rota-Baxter operator on $(A, \searrow, \nearrow, \nwarrow,
\swarrow)$. Then there exists an octo-algebra structure on $A$
defined by (for any $x,y\in A$)
$$x\searrow_1 y= x\searrow R(y),\;x\searrow_2 y=R(x)\searrow y,\;
x\nearrow_1 y= x\nearrow R(y),\;x\nearrow_2 y=R(x)\nearrow y,\;$$
$$x\nwarrow_1 y= x\nwarrow R(y),\;x\nwarrow_2 y=R(x)\nwarrow y,\;
x\swarrow_1 y= x\swarrow R(y),\;x\swarrow_2 y=R(x)\swarrow y.\eqno
(4.4.20)$$\end{coro}

\begin{lemma} {\rm (\cite{AL}, Proposition 2.5 and \cite{Le3}, Proposition
5.13)}\quad Let $R_1,R_2$ and $R_3$ be three pairwise commutating
Rota-Baxter operators on an associative algebra $(A,*)$. Then $R_2$
is a Rota-Baxter operator on the dendriform algebra obtained from
$R_3$ by Corollary 2.1.9. Moreover, $R_1$ is a Rota-Baxter operator
on the quadri-algebra obtained from the above dendriform algebra and
$R_2$ by Corollary 3.4.8.\end{lemma}


Therefore by Corollaries 2.1.9, 3.4.8 and 4.4.8, it is obvious that
the Rota-Baxter operators on associative algebras also can construct
octo-algebras as follows.

\begin{coro} Let $R_1, R_2$ and $R_3$ be three pairwise commuting
Rota-Baxter operators on an associative algebra $(A,*)$. Then there
exists an octo-algebra structure on $A$ defined by
$$x\searrow_1y=R_2R_3(x)*R_1(y),\; x\searrow_2 y=R_1R_2R_3(x)*y,\;
x\nearrow_1y=R_2(x)*R_1R_3(y),$$$$x\nearrow_2y=R_1R_2(x)*R_3(y),
x\swarrow_1y=R_3(x)*R_1R_2(y),\;x\swarrow_2y=R_1R_3(x)*R_2(y),$$
$$x\nwarrow_1y=x*R_1R_2R_3(y),\;x\nwarrow_2
y=R_1(x)*R_2R_3(y),\;\;\forall x,y\in A.\eqno (4.4.21)$$\end{coro}

\begin{coro} Let $(A, \searrow, \nearrow, \nwarrow, \swarrow)$ be a
quadri-algebra. Then there exists a compatible octo-algebra
structure on $A$ such that $(A, \searrow, \nearrow, \nwarrow,
\swarrow)$ is the associated depth quadri-algebra if and only if
there exists an invertible $\mathcal O$-operator $T$ of $(A,
\searrow, \nearrow, \nwarrow$, $\swarrow)$ associated to certain
bimodule $(l_\searrow, r_\searrow, l_\nearrow, r_\nearrow,
l_\nwarrow,r_\nwarrow, l_\swarrow, r_\swarrow,V)$ (hence $\dim
V=\dim A$).\end{coro}

\begin{proof} The compatible octo-algebra structure on $(A, \searrow, \nearrow, \nwarrow, \swarrow)$ is given
by
$$x\searrow_1 y=T(r_\searrow(y)T^{-1}(x)),\;
x\searrow_2 y=T(l_\searrow(x)T^{-1}(y)),\;x\nearrow_1
y=T(r_\nearrow(y)T^{-1}(x)),$$ $$ x\nearrow_2
y=T(l_\nearrow(x)T^{-1}(y)),\;x\swarrow_1
y=T(r_\swarrow(y)T^{-1}(x)),\; x\swarrow_2
y=T(l_\swarrow(x)T^{-1}(y)),$$
$$x\nwarrow_1 y=T(r_\nwarrow(y)T^{-1}(x)),\;
x\nwarrow_2 y=T(l_\nwarrow(x)T^{-1}(y)),\;\;\forall x,y\in A.$$
Conversely, let  $(A, \searrow_1, \searrow_2, \nearrow_1,
\nearrow_2, \nwarrow_1, \nwarrow_2, \swarrow_1, \swarrow_2)$ be an
octo-algebra and $(A,\searrow_{12},\nearrow_{12},\nwarrow_{12}$,
$\swarrow_{12})$ be the associated depth quadri-algebra. Then
$(L_{\searrow_2}, R_{\searrow_1}, L_{\nearrow_2}, R_{\nearrow_1},
L_{\nwarrow_2}, R_{\nwarrow_1},L_{\swarrow_2}$, $R_{\swarrow_1}$,
$ A)$ is a bimodule of
$(A,\searrow_{12},\nearrow_{12},\nwarrow_{12},\swarrow_{12})$ and
the identity map $id$ is an $\mathcal O$-operator of $(A,
\searrow$, $\nearrow$, $\nwarrow, \swarrow)$ associated to it.
\end{proof}

\begin{prop} Let $(A, \searrow, \nearrow, \nwarrow, \swarrow)$ be a
quadri-algebra. If there is a nondegenerate skew-symmetric 2-cocycle
$\omega$ of $(A, \searrow, \nearrow, \nwarrow, \swarrow)$, then
there exists a compatible octo-algebra structure on $(A, \searrow,
\nearrow, \nwarrow, \swarrow)$  defined by (for any $x,y,z\in A$)
$$\omega(x\searrow_1 y,z)=\omega(x,y\nwarrow
z),\;\;\omega(x\searrow_2 y,z)=\omega (y,z*x),$$
$$\omega(x\nearrow_1 y,z)=\omega(x,-y\prec
z),\;\;\omega(x\searrow_2 y,z)=\omega (y,-z\vee x),$$
$$\omega(x\nwarrow_1 y,z)=\omega(x,y*
z),\;\;\omega(x\nwarrow_2 y,z)=\omega (y,z\searrow x),$$
$$\omega(x\swarrow_1 y,z)=\omega(x,-y\wedge
z),\;\;\omega(x\swarrow_2 y,z)=\omega (y,-z\succ x),\eqno
(4.4.22)$$ such that $(A, \searrow, \nearrow, \nwarrow, \swarrow)$
is the associated depth quadri-algebra.
\end{prop}

\begin{proof} The proof is similar as of Proposition 3.4.11.
Here the invertible linear map $T:A\rightarrow A^*$ is defined by
$\langle T(x),y\rangle=\omega (x,y)$. Then $T^{-1}$ is an invertible
$\mathcal O$-operator of $(A, \searrow, \nearrow, \nwarrow,
\swarrow)$ associated to the bimodule $(R_*^*, L_\nwarrow^*$,
$-R_\vee^*,
-L_\prec^*,R_\searrow^*,L_*^*,-R_\succ^*,-L_\wedge^*,A^*)$.
\end{proof}

\begin{coro}{\rm (cf. Corollary 4.2.10)}\quad Let  $(A, \searrow_1, \searrow_2,
\nearrow_1, \nearrow_2, \nwarrow_1, \nwarrow_2, \swarrow_1,
\swarrow_2)$ be an octo-algebra and
$(A,\searrow_{12},\nearrow_{12},\nwarrow_{12},\swarrow_{12})$ be
the associated depth quadri-algebra. Then $r$ given by equation
(2.2.9) is a skew-symmetric solution of $Q$-equation in the
quadri-algebra
$$A\ltimes_{R_{\Sigma_1}^*, L_{\nwarrow_{2}}^*,-R_{\vee_1}^*,
-L_{\prec_2}^*,R_{\searrow_{1}}^*,L_{\Sigma_2}^*,-R_{\succ_1}^*,-L_{\wedge_2}^*}A^*,$$
where $\{e_1,\cdots, e_n\}$ is a basis of $A$ and $\{e_1^*,\cdots,
e_n^*\}$ is its dual basis. Moreover there is a natural 2-cocycle of
the quadri-algebra $A\ltimes_{R_{\Sigma_1}^*,
L_{\nwarrow_{2}}^*,-R_{\vee_1}^*,
-L_{\prec_2}^*,R_{\searrow_{1}}^*,L_{\Sigma_2}^*,-R_{\succ_1}^*,-L_{\wedge_2}^*}A^*,$
induced by $r^{-1}: A\oplus A^*\rightarrow (A\oplus A^*)^*$, which
is given by equation (2.2.10).\end{coro}

\begin{proof}
The proof is similar as of Proposition 3.4.12. Note that here
$$(R_{\Sigma_1}^*, L_{\nwarrow_{2}}^*,-R_{\vee_1}^*,
-L_{\prec_2}^*,R_{\searrow_{1}}^*,L_{\Sigma_2}^*,-R_{\succ_1}^*,-L_{\wedge_2}^*,A^*)$$
is the dual bimodule of the bimodule $(L_{\searrow_2},
R_{\searrow_1}, L_{\nearrow_2}, R_{\nearrow_1}, L_{\nwarrow_2},
R_{\nwarrow_1},L_{\swarrow_2}, R_{\swarrow_1},  A)$
 of the
associated depth quadri-algebra
$(A,\searrow_{12},\nearrow_{12},\nwarrow_{12},\swarrow_{12})$ and
 ${id}$ is an ${\mathcal O}$-operator of
$(A,\searrow_{12},\nearrow_{12}$, $\nwarrow_{12},\swarrow_{12})$
associated to the bimodule $(L_{\searrow_2}, R_{\searrow_1},
L_{\nearrow_2}, R_{\nearrow_1}, L_{\nwarrow_2},
R_{\nwarrow_1},L_{\swarrow_2}, R_{\swarrow_1},  A)$.
\end{proof}


\begin{prop}Let  $(A, \searrow_1, \searrow_2, \nearrow_1, \nearrow_2,
\nwarrow_1, \nwarrow_2, \swarrow_1, \swarrow_2)$ be an
octo-algebra and $(A$, $\searrow_{12}$,
$\nearrow_{12},\nwarrow_{12},\swarrow_{12})$ be the associated
depth quadri-algebra. Let $r\in A\otimes A$ be symmetric. Then $r$
is an ${\mathcal O}$-operator of
$(A,\searrow_{12},\nearrow_{12},\nwarrow_{12},\swarrow_{12})$
associated to the bimodule
 $$(R_{\Sigma_1}^*,
L_{\nwarrow_{2}}^*,-R_{\vee_1}^*,
-L_{\prec_2}^*,R_{\searrow_{1}}^*,L_{\Sigma_2}^*,-R_{\succ_1}^*,-L_{\wedge_2}^*,A^*)$$
if and only if $r$ satisfies
$$r_{13}\searrow r_{23}=r_{23}\Sigma_1 r_{12}+r_{12}\nwarrow_2
r_{13},\eqno (4.4.23)$$
$$r_{13}\nearrow r_{23}=-r_{23}\vee_1 r_{12}-r_{12}\prec_2
r_{13},\eqno (4.4.24)$$
$$r_{13}\nwarrow r_{23}=r_{23}\searrow_1 r_{12}+r_{12}\Sigma_2
r_{13},\eqno (4.4.25)$$
$$r_{13}\swarrow r_{23}=-r_{23}\succ_1 r_{12}-r_{12}\wedge_2
r_{13}.\eqno (4.4.26)$$\end{prop}

\begin{proof}
The proof is similar as of Proposition 3.4.13.
\end{proof}

\begin{remark}{\rm From the study in the previous sections, it is enough
reasonable to regard a set of equations (4.4.23)-(4.2.26) in an
octo-algebra as an analogue of the classical Yang-Baxter equation
(\cite{NB}). We call it {\it $O$-equation} in an
octo-algebra.}\end{remark}

\section{Summary and generalization}

\subsection{Summary}
We can summarize the study in the previous sections as follows.

(I) Algebra structures. The relations between the different algebras
appearing in this paper can be summarized by the following diagram.
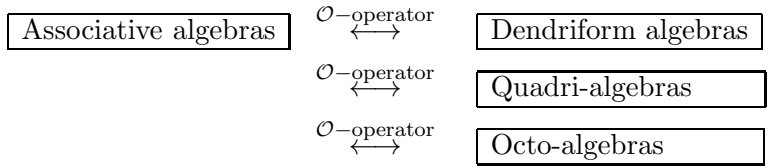
\begin{eqnarray*}
\begin{tabular}{|c|}\hline
{Associative algebras}\\\hline
\end{tabular}
 &\stackrel{\mathcal O-{\rm operator}}{\longleftrightarrow}&\;\;
\begin{tabular}{|c|}\hline
{Dendriform algebras}\\\hline
\end{tabular}\\
&\stackrel{\mathcal O-{\rm operator}}{\longleftrightarrow}&\;\;
\begin{tabular}{|c|}\hline {Quadri-algebras}\hspace{0.8cm}\mbox{}\\\hline
\end{tabular}\\
&\stackrel{\mathcal O-{\rm operator}}{\longleftrightarrow}&\;\;
\begin{tabular}{|c|}\hline {Octo-algebras}\hspace{1.15cm}\mbox{}\\\hline
\end{tabular}
\end{eqnarray*}
Here the meaning of $\stackrel{\mathcal O-{\rm
operator}}{\longleftrightarrow}$ is given as follows.

(a) $\stackrel{\mathcal O-{\rm operator}}{\longrightarrow}$: An
arrow-ending algebra can be obtained from an arrow-starting algebra
by an $\mathcal O$-operator of the arrow-starting algebra (see
Theorem 2.1.8, Proposition 3.4.6 and Proposition 4.4.6). As special
cases, the construction of the arrow-ending algebras from the
arrow-starting algebras by the Rota-Baxter operators follows
immediately (\cite{Ag2},\cite{AL},\cite{Le3}).

(b) $\stackrel{\mathcal O-{\rm operator}}{\longleftarrow}$: The
existence of a compatible arrow-starting algebra on an arrow-ending
algebra is decided by the existence of an invertible $\mathcal
O$-operator of the arrow-ending algebra (see Corollary 2.1.10,
Corollary 3.4.10, Corollary 4.4.11). In particular, $id$ is an
$\mathcal O$-operator of the arrow-ending algebra associated to
certain bimodule given by the multiplication operators (or
equivalently, the definition identities) of the arrow-starting
algebras (see Corollary 2.1.7, Proposition 3.4.3 and Proposition
4.4.3).

(II) Algebraic equations (analogues of the classical Yang-Baxter
equation) and bilinear forms. There is a ``chain" of algebraic
equations and bilinear forms on the associative algebras, dendriform
algebras and quadri-algebras corresponding to the algebra relations
given in the above (I). It can be interpreted explicitly by the
following diagrams.

\vspace{0.5cm}

{\small $\begin{matrix}
\begin{tabular}{|c|}\hline
\mbox{}\hspace{0.38cm}{\bf Associative
algebra}\hspace{0.38cm}\mbox{}\\\hline Symmetric bilinear
form\\\hline Invariant\\\hline
\end{tabular}&&&&\cr
&&&&\cr \updownarrow\;{\rm  D.C.}&&&&\cr &&&&\cr
\begin{tabular}{|c|}\hline
{\bf Associative algebra}\\\hline Associative Yang-Baxter\\
equation\\\hline Skew-symmetric solution\\\hline Skew-symmetric part
of an\\ $\mathcal O$-operator\\\hline $\mathcal O$-operator
associated to\\ $(R^*, L^*)$\\\hline
\end{tabular} &\Longrightarrow&
\begin{tabular}{|c|}\hline
{\bf Associative algebra}\\\hline Skew-symmetric bilinear
form\\\hline Connes cocycle\\\hline
\end{tabular} & \Longrightarrow &
\begin{tabular}{|c|}\hline
Dendriform algebra\\\hline
\end{tabular}\cr
& &\updownarrow\;{\rm D.C.}&&\cr &&&&\cr
&&\begin{tabular}{|c|}\hline \mbox{}\hspace{0.6cm}{\bf
Associative algebra}\hspace{0.6cm}\mbox{}\\\hline $D$-equation in a dendriform\\
algebra\\\hline Symmetric solution\\\hline Symmetric part of an\\
$\mathcal O$-operator\\\hline $\mathcal O$-operator associated to\\
$(R_\prec^*, L_\succ^*)$\\\hline
\end{tabular}&&\cr
&&&&\cr {\rm Part\;\; (AI)} && {\rm Part\;\; (AII)}&& {\rm
Part\;\;(AIII)}
\end{matrix}$}

\bigskip

{\small $
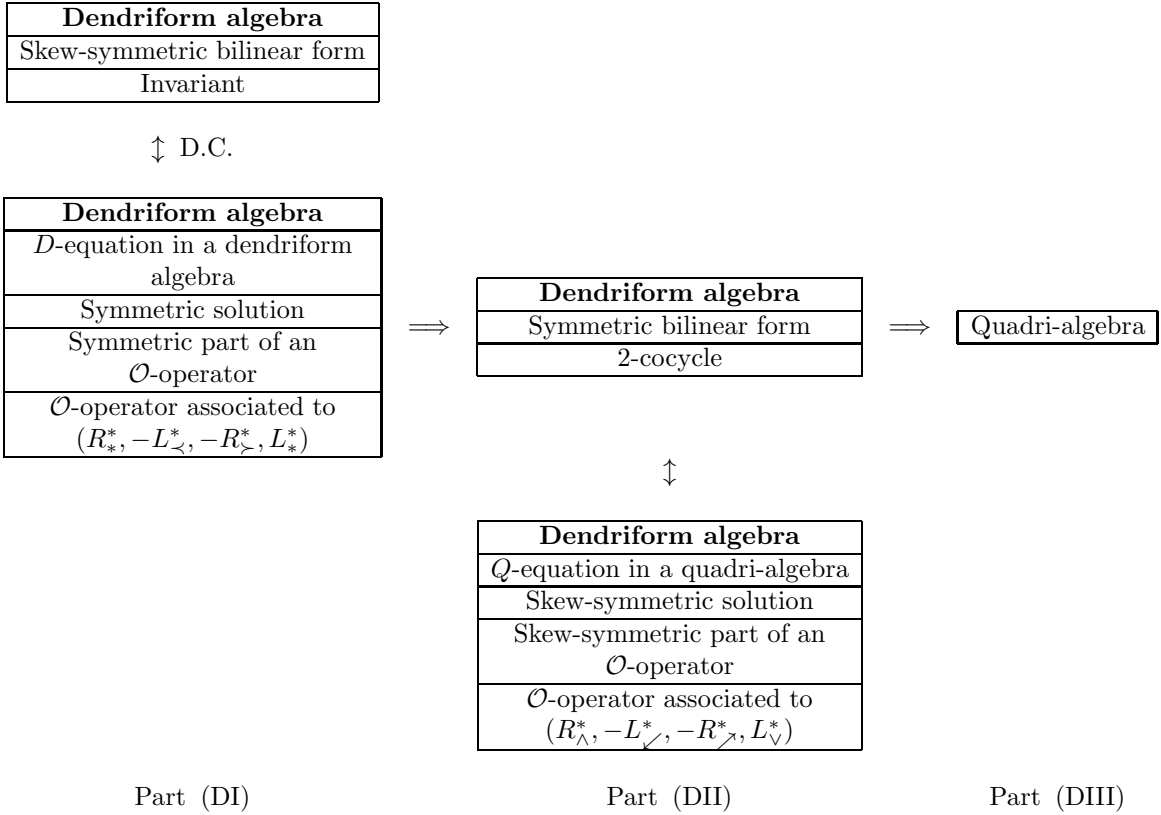
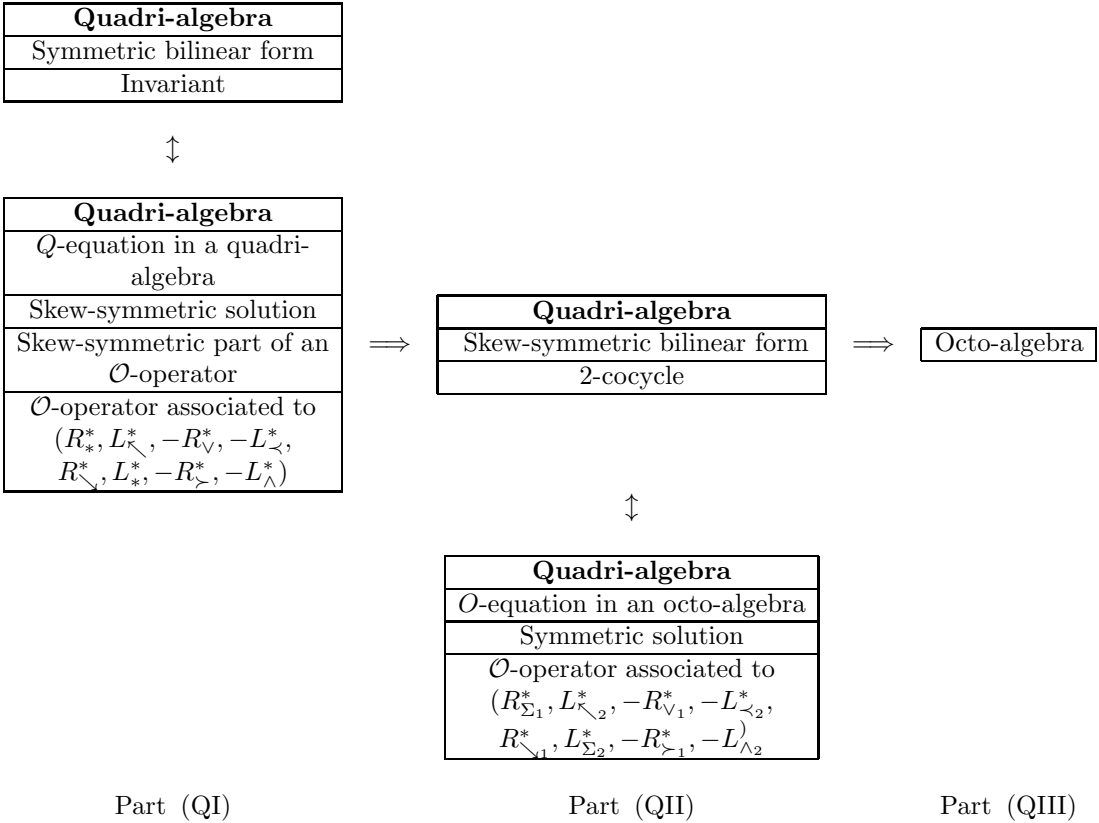
\begin{matrix}
\begin{tabular}{|c|}\hline
\mbox{}{\bf Dendriform algebra}\mbox{}\\\hline Skew-symmetric
bilinear form\\\hline Invariant \\\hline
\end{tabular}&&&&\cr&&&&\cr
\updownarrow\;{\rm  D.C.}&&&&\cr&&&&\cr
\begin{tabular}{|c|}\hline
\mbox{}\hspace{0.6cm}{\bf Dendriform algebra}\hspace{0.6cm}\mbox{}\\\hline $D$-equation in a dendriform\\
algebra\\\hline Symmetric solution\\\hline Symmetric part of an\\
$\mathcal O$-operator\\\hline $\mathcal O$-operator associated to\\
$(R_*^*,-L_\prec^*, -R_\succ^*, L_*^*)$\\\hline
\end{tabular} &\Longrightarrow&
\begin{tabular}{|c|}\hline
\mbox{}\hspace{0.65cm}{\bf Dendriform
algebra}\hspace{0.65cm}\mbox{}\\\hline Symmetric bilinear
form\\\hline 2-cocycle
\\\hline
\end{tabular} & \Longrightarrow &
\begin{tabular}{|c|}\hline
Quadri-algebra\\\hline
\end{tabular}\cr
& &\updownarrow&&\cr &&&&\cr &&\begin{tabular}{|c|}\hline {\bf
Dendriform algebra}\\\hline $Q$-equation in a quadri-algebra\\\hline
Skew-symmetric solution\\\hline Skew-symmetric part of an
\\$\mathcal O$-operator\\\hline $\mathcal O$-operator associated
to\\
$(R_\wedge^*, -L_\swarrow^*,-R_\nearrow^*, L_\vee^*)$\\\hline
\end{tabular}&&\cr
&&&&\cr {\rm Part\;\; (DI)} && {\rm Part\;\; (DII)}&& {\rm
Part\;\;(DIII)}
\end{matrix}$}

\bigskip

{\small $\begin{matrix}
\begin{tabular}{|c|}\hline
\mbox{}\hspace{0.75cm}{\bf
Quadri-algebra}\hspace{0.75cm}\mbox{}\\\hline Symmetric bilinear
form\\\hline Invariant \\\hline
\end{tabular}&&&&\cr&&&&\cr
\updownarrow&&&&\cr&&&&\cr
\begin{tabular}{|c|}\hline
{\bf Quadri-algebra}\\\hline $Q$-equation in a quadri-\\
algebra\\\hline Skew-symmetric solution\\\hline Skew-symmetric part of an\\
$\mathcal O$-operator\\\hline $\mathcal O$-operator associated to\\
$(R_*^*, L_\nwarrow^*, -R_\vee^*,-L_\prec^*,$\\ $R_\searrow^*,
L_*^*, -R_\succ^*, -L_\wedge^*)$\\\hline
\end{tabular} &\Longrightarrow&
\begin{tabular}{|c|}\hline
{\bf Quadri-algebra}\\\hline\mbox{}\hspace{0.05cm} Skew-symmetric
bilinear form\hspace{0.05cm}\mbox{}\\\hline 2-cocycle
\\\hline
\end{tabular} & \Longrightarrow &
\begin{tabular}{|c|}\hline
Octo-algebra\\\hline
\end{tabular}\cr
& &\updownarrow&&\cr &&&&\cr &&\begin{tabular}{|c|}\hline {\bf
Quadri-algebra}\\\hline $O$-equation in an octo-algebra\\\hline
Symmetric solution\\\hline $\mathcal O$-operator associated
to\\
$(R_{\Sigma_1}^*, L_{\nwarrow_{2}}^*,-R_{\vee_1}^*,
-L_{\prec_2}^*,$\\
$R_{\searrow_{1}}^*,L_{\Sigma_2}^*,-R_{\succ_1}^*,-L_{\wedge_2}^)$\\\hline
\end{tabular}&&\cr
&&&&\cr {\rm Part\;\; (QI)} && {\rm Part\;\; (QII)}&& {\rm
Part\;\;(QIII)}
\end{matrix}$}

\mbox{}

 \noindent Here ``D.C.'' is the abbreviation of ``double
construction" and ``$\Longrightarrow$'' is under the invertible
cases. Moreover, we have the following equivalences.
$${\rm Part}\;\;{\rm (AII)}\;\;\Longleftrightarrow\;\; {\rm Part}\;\;{\rm
(DI)}; \quad \quad {\rm Part}\;\;{\rm
(DII)}\;\;\Longleftrightarrow\;\; {\rm Part}\;\;{\rm (QI)}.$$

A simplified illustration can be expressed as follows.

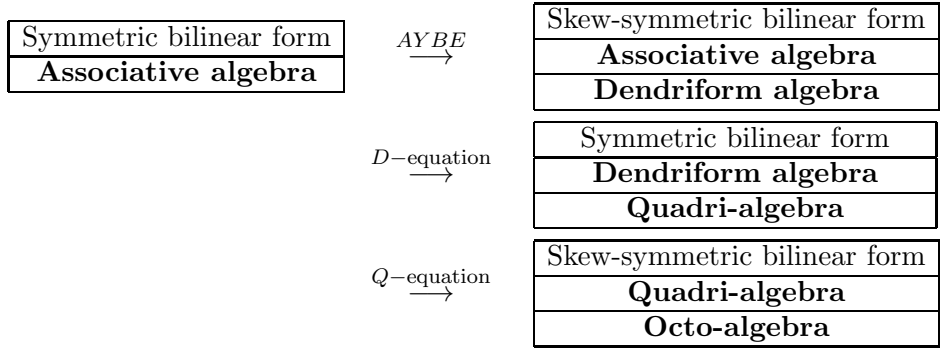
\begin{eqnarray*}
\begin{tabular}{|c|}\hline
Symmetric bilinear form\\\hline {\bf Associative algebra}\\\hline
\end{tabular}
 &\stackrel{AYBE}{\longrightarrow}&\;\;
\begin{tabular}{|c|}\hline
Skew-symmetric bilinear form\\\hline {\bf Associative
algebra}\\\hline {\bf Dendriform algebra}\\\hline
\end{tabular}\\
&\stackrel{D-{\rm equation}}{\longrightarrow}&\;\;
\begin{tabular}{|c|}\hline Symmetric bilinear form\\\hline {\bf Dendriform
algebra}\\\hline\mbox{}\hspace{1.05cm}{\bf
Quadri-algebra}\hspace{1.05cm}\mbox{}\\\hline
\end{tabular}\\
&\stackrel{Q-{\rm equation}}{\longrightarrow}&\;\;
\begin{tabular}{|c|}\hline
Skew-symmetric bilinear form\\\hline {\bf Quadri-algebra}\\\hline
{\bf Octo-algebra}\\\hline
\end{tabular}
\end{eqnarray*}

\subsection{Generalization}

It is of course natural to continue and extend the study in the
previous sections to octo-algebras (to study their $\mathcal
O$-operators on and the related topics) and even the algebra systems
with more operations in an associative cluster. In fact, we can give
an outline of such a study by induction. Thus, starting from the
associative algebras, a similar study as in this paper can be easily
given for any algebra with $2^k$ operations in  an associative
cluster.

Let $(A,\circ_1,\circ_2,\cdots,\circ_{2^n})$ be an algebra with
$2^n$ operations in an associative cluster. Suppose that there has
already been a similar study as follows on the algebras with $2^m$
operations for any $m<n$. In particular, a bilinear form on $A$
satisfying certain conditions with fixed symmetry (see the following
step (5)) is known due to our assumption that the theory starts from
the associative algebras with symmetric invariant bilinear forms.
Without loss of generality, we assume that it is symmetric. Then the
whole study can be divided into 6 steps and the explicit procedure
is given as follows.

Step (1)\quad Give the bimodule structures of
$(A,\circ_1,\circ_2,\cdots,\circ_{2^n})$. It can be done by the way
of semidirect sum (\cite{Sc}). It is easy to get the relations of
between them and the bimodules of the algebras with $2^m$ operations
for any $m<n$. Then the dual bimodule structure (see the proof of
Proposition 4.1.4) can be found.

Step (2)\quad Take the $2^{n+1}$ operations on $A$ denoted by
$$\circ_1^1,\circ_1^2,\circ_2^1,\circ_2^2,\cdots,\circ_{2^n}^1,\circ_{2^n}^2:A\otimes
A\rightarrow A.$$ The definition identity of the algebra
$(A,\circ_1^1,\circ_1^2,\circ_2^1,\circ_2^2,\cdots,\circ_{2^n}^1,\circ_{2^n}^2)$
is decided by
$$x\circ_i y=x\circ_i^1 y+x\circ_i^2 y,\;\;\forall x,y\in A$$
and $(L_{\circ_1^2},R_{\circ_1^1},\cdots, L_{\circ_{2^n}^2},
R_{\circ_{2^n}^1}, A)$ is a bimodule of
$(A,\circ_1,\circ_2,\cdots,\circ_{2^n})$.

Step (3)\quad Define the $\mathcal O$-operators of
$(A,\circ_1,\circ_2,\cdots,\circ_{2^n})$ by a natural and obvious
way. Then the relations between the two algebras, in particular, the
construction of $(A,\circ_1^1,\circ_1^2,\circ_2^1,\circ_2^2,\cdots$,
$\circ_{2^n}^1,\circ_{2^n}^2)$ from
$(A,\circ_1,\circ_2,\cdots,\circ_{2^n})$ with an $\mathcal
O$-operator of the latter algebra and the existence of the former
algebra in the latter algebra are given as we have summarized in (I)
in subsection 5.1.

Step (4)\quad The skew-symmetric $\mathcal O$-operator of
$(A,\circ_1,\circ_2,\cdots,\circ_{2^n})$ associated to the dual
bimodule of the regular bimodule $(L_{\circ_1},R_{\circ_1},\cdots,
L_{\circ_{2^n}},R_{\circ_{2^n}}, A)$ of
$(A,\circ_1,\circ_2,\cdots,\circ_{2^n})$ gives an explicit algebraic
equation (or a set of equations) in
$(A,\circ_1,\circ_2,\cdots,\circ_{2^n})$, which can be regarded as
an analogue of the classical Yang-Baxter equation.

Step (5)\quad The skew-symmetric invertible solution of the equation
obtained in the step (4) gives a skew-symmetric bilinear form
satisfying certain conditions. Moreover, such a nondegenerate
skew-symmetric bilinear form can induce an algebra
$(A,\circ_1^1,\circ_1^2,\circ_2^1,\circ_2^2,\cdots,\circ_{2^n}^1,\circ_{2^n}^2)$
and the skew-symmetric bilinear form is just what we try to find for
both
$$(A,\circ_1,\circ_2,\cdots,\circ_{2^n})\;\;{\rm
and}\;\;(A,\circ_1^1,\circ_1^2,\circ_2^1,\circ_2^2,\cdots,\circ_{2^n}^1,\circ_{2^n}^2).$$

Step (6)\quad The symmetric $\mathcal O$-operator of
$(A,\circ_1,\circ_2,\cdots,\circ_{2^n})$ associated to the dual
bimodule of the bimodule
$$(L_{\circ_1^2},R_{\circ_1^1},\cdots, L_{\circ_{2^n}^2},
R_{\circ_{2^n}^1}, A)$$ can give an algebraic equation (or a set of
equations) in
$(A,\circ_1^1,\circ_1^2,\circ_2^1,\circ_2^2,\cdots,\circ_{2^n}^1,\circ_{2^n}^2)$
which can also be regarded as an analogue of the classical
Yang-Baxter equation (it coincides with the equation given in the
step (4) for the algebra
$(A,\circ_1^1,\circ_1^2,\circ_2^1,\circ_2^2,\cdots,\circ_{2^n}^1,\circ_{2^n}^2)$).

It is also reasonable to give the following conjecture, which may
indicate an application of the analogues of the classical
Yang-Baxter equation.

{\bf Conjecture}\quad There should be a similar double construction
of the nondegenerate bilinear forms corresponding to the algebraic
equation given in the above step (4) for any algebra
$(A,\circ_1,\circ_2,\cdots,\circ_{2^n})$ as we have shown in
subsections 2.2, 2.3 and 3.2 for associative algebras and dendriform
algebras (\cite{Bai3}). It has been proved to be true for
quadri-algebras (\cite{Ni}) and octo-algebras (\cite{NB}).

As we have pointed out in the Introduction, it is likely to give an
operadic interpretation for the study in this paper and the further
development on this subject.

\section*{Acknowledgements}

The author thanks Professors M. Aguiar, L. Guo and the referee for
important suggestion. This work was supported in part by the
National Natural Science Foundation of China (10621101), NKBRPC
(2006CB805905), SRFDP (200800550015).

\end{document}